\documentclass[letterpaper]{scrartcl}
  
\usepackage[totalwidth=430pt, totalheight=600pt]{geometry}

\makeatletter
\def\@fnsymbol#1{\ensuremath{\ifcase#1\or * \or \ddagger\or
   ~ \or \mathparagraph\or \|\or **\or \dagger\dagger
   \or \ddagger\ddagger \else\@ctrerr\fi}}
    \makeatother

\usepackage[english]{babel}
    
\usepackage{amsfonts}
\usepackage{amsmath}
\usepackage{amsthm}
\usepackage{microtype}
\usepackage{diagbox}

\usepackage[font=small]{caption}
\usepackage{comment}
\usepackage{enumerate}
\usepackage{xfrac}
\usepackage{noindentafter}
\usepackage{stmaryrd}

\usepackage{thm-restate}

\usepackage[hidelinks]{hyperref}
\usepackage{cleveref}

\usepackage{tikz}
\usetikzlibrary{shapes.misc}

\newcommand{\N}{\mathbb{N}}
\newcommand{\R}{\mathbb{R}}

\newcommand{\fL}{\mathcal{L}} %
\newcommand{\fO}{\mathcal{O}} %
\newcommand{\fP}{\mathcal{P}} %
\newcommand{\fQ}{\mathcal{Q}} %
\newcommand{\fR}{\mathcal{R}} %
\newcommand{\fS}{\mathcal{S}}
\newcommand{\rD}{\mathrm{D}}
\newcommand{\rS}{\mathrm{S}}
\newcommand{\interior}{\mathrm{interior}}
\newcommand{\rect}{\mathrm{rect}}
\newcommand{\ceil}[1]{\lceil #1 \rceil}

\newtheorem{lemma}{Lemma} %
\newtheorem{proposition}{Proposition} %

\newtheorem{corollary}{Corollary}

\theoremstyle{definition}

\newtheorem{openQuestion}{Open question}

\NoIndentAfterEnv{lemma}
\NoIndentAfterEnv{proposition}
\NoIndentAfterEnv{theorem}
\NoIndentAfterEnv{conjecture}
\NoIndentAfterEnv{corollary}
\NoIndentAfterEnv{observation}
\NoIndentAfterEnv{definition}
\NoIndentAfterEnv{openQuestion}

\NoIndentAfterEnv{restatable}
\NoIndentAfterEnv{proof}

\NoIndentAfterEnv{align*}
\NoIndentAfterEnv{enumerate}
\NoIndentAfterEnv{itemize}

\title{{Geometric group testing}}

\author{Benjamin Aram Berendsohn\thanks{Institut f\"ur Informatik, Freie Universit\"at Berlin, \texttt{beab@zedat.fu-berlin.de}} \footnotemark[3] \and L\'aszl\'o Kozma\thanks{Institut f\"ur Informatik, Freie Universit\"at Berlin, \texttt{laszlo.kozma@fu-berlin.de}} \footnote{Work supported by DFG grant KO 6140/1-1.}}%

\begin{document}
\date{}
	\maketitle
	
	\begin{abstract}
		Group testing is concerned with identifying $t$ defective items in a set of $m$ items, where each test reports whether a specific subset of items contains at least one defective. In \emph{non-adaptive} group testing, the subsets to be tested are fixed in advance. By testing multiple items at once, the required number of tests can be significantly smaller than $m$. In fact, for $t \in \fO(1)$, the optimal number of (non-adaptive) tests is known to be $\Theta(\log{m})$.
		
		In this paper, we consider the problem of non-adaptive group testing in a \emph{geometric setting}, where the items are points in $d$-dimensional Euclidean space and the tests are axis-parallel boxes (hyperrectangles).  
		We present upper and lower bounds on the required number of tests under this geometric constraint. In contrast to the general, combinatorial case, the bounds in our geometric setting are polynomial in $m$. For instance, our results imply that identifying a defective pair in a set of $m$ points in the plane always requires $\Omega(m^{3/5})$ tests, and there exist configurations of $m$ points for which $\fO(m^{2/3})$ tests are sufficient, whereas to identify a single defective point in the plane, $\Theta(m^{1/2})$ tests are always necessary and sometimes sufficient. 
	\end{abstract}
	
	\section{Introduction}\label{sec:intro}

Group testing is often introduced through the following familiar example. Given a collection of $m$ lightbulbs, one of which is defective, the task is to identify, with as few tests as possible, the defective lightbulb. A test consists of connecting a number of lightbulbs in series with a power source, thereby detecting whether the defective lightbulb is in the tested group. The na\"ive method of separately testing each lightbulb requires $m$ tests. By testing multiple items at once, we may be able to identify the defective with significantly fewer tests.

More generally, consider a collection of $m$ objects, of which $t$ are defective, where $t$ is known. The goal is to identify the $t$ defective objects with as few tests as possible. A test is performed on a selected subset of objects. The outcome of a test is \emph{positive} if at least one of the selected objects is defective, and \emph{negative} otherwise. In \emph{non-adaptive} group testing, the tests to be performed are chosen in advance, %
without access to the results of other tests. %
In this paper we only concern ourselves with non-adaptive testing.\footnote{A natural, closely related variant of the problem is where the number $t$ of defectives is not known, and only an upper bound $t' \geq t$ is given. We discuss this setting in Appendix~\ref{appb}.}

Group testing was introduced by Dorfman~\cite{dorfman1943} in the 1940s with medical applications in sight (e.g.\ pooling together multiple blood samples when testing for infections). Several aspects of group testing have since been thoroughly studied, and group testing has been applied in various fields (we refer to the comprehensive textbook of Du and Hwang~\cite{DuHwang1993} and the recent survey of Aldridge, Johnson, and Scarlett~\cite{AJS}). Non-adaptive group testing has close connections with error-correcting-codes~\cite{Indyk10, Porat08}, combinatorial designs~\cite[\S\,11]{Design_book}, compressive sensing~\cite{gilbert2008group, Gilbert10, Gilbert12, Ngo12, emad14}, streaming algorithms~\cite[\S\,6.1]{Muthu05}, \cite{Cormode05}, inference and learning~\cite{Ubaru, Zhou, malioutov13,simsearch}.

It is well-known that to identify $t \in \fO(1)$ defective items in a set of $m$ items, $\Theta(\log{m})$ non-adaptive tests are necessary and sufficient (see~\cite[Theorems 7.2.12 and 7.2.15]{DuHwang1993} for bounds with explicit dependence on $t$). For this result to hold, it is assumed that arbitrary subsets of the items can be selected as tests. In some applications this assumption may not be realistic. For instance, in the example mentioned in the beginning, it may be the case that all lightbulbs are initially connected in series and during testing we can not change the wiring, apart from connecting a power source at two arbitrary points. In this case, tests are in effect restricted to groups of items contained in contiguous intervals.

In this paper we study non-adaptive group testing under more general geometric constraints. Our items to be tested are assumed to be points in $d$-dimensional Euclidean space, and our tests are subsets induced by $d$-dimensional axis-parallel boxes (i.e.\ hyperrectangles). Point sets in high-dimensional space are commonly used to model databases with multivariate numerical data, and axis-parallel boxes naturally map to \emph{orthogonal range queries} (i.e.\ queries where each numerical parameter is restricted to some interval).

We give upper and lower bounds on the number of tests necessary to detect defectives in this geometric setting, and observe that these bounds are polynomial in $m$, in stark contrast to the general case.
For convenience, we state and prove our results in an equivalent setting, where the number of tests ($n$) is fixed  and the number of items ($m$) is to be \emph{maximized}. Upper bounds on $m$ hold for arbitrary configurations of points and tests, and thus show the inherent limitations of group-testing with orthogonal range queries. 

Lower bounds on $m$ are concrete feasible configurations of points and tests in which group testing can be performed efficiently. %
In general, even improving upon the na\"ive ``one test per item'' turns out to be a non-trivial task, and we consider finding optimal configurations an interesting extremal geometric question. Note that if the placement of points is \emph{adversarial}, then only the trivial lower bound $m = \Omega(n)$ is possible, even in the case of a single defective point. (A possible worst-case construction is to place all $m$ points on an axis-parallel line in $\R^d$, in which case $n$ tests can handle only $O(n)$ points.)

With more benign point placements, however, asymptotic improvements over the worst-case are possible. %
We illustrate this with a simple example: consider $n^2$ points in the $2$-dimensional plane, with a single defective, i.e.\ $t=1$. Place the points on the $n \times n$ grid, and place $2n$ rectangle-tests, such as to cover each row and column of the grid by a unique rectangle (see \Cref{fig:long_rect_step}~(a)). Observe that exactly two rectangle-tests will evaluate positive, identifying the defective item by its two coordinates. We argue later (Proposition~\ref{p:rect-trivial-upper-bound}) that this configuration is essentially optimal. Observe however, that the same configuration would not be able to identify \emph{two} defective items, unless they were on the same row or column of the grid.

Configurations of items and tests can be described as set systems.\footnote{In the group testing literature an equivalent matrix notation is often used.} Two properties of set systems are central in non-adaptive group testing: \emph{$t$-separability} and \emph{$t$-disjunctness}. %

	\subsection{\texorpdfstring{$t$-separable and $t$-disjunct}{t-separable and t-disjunct} set systems}
	\label{sec11}
	Let $(X,\fS)$ be a set system, where $X$ is a finite set and $\fS \subseteq 2^X$. We sometimes refer to elements of $X$ as \emph{items} and to elements of $\fS$ as \emph{tests}. For $x \in X$, let $\fS[x] = \{ S \subseteq \fS \mid {x \in S} \}$ denote the set of tests that contain $x$. For $Y \subseteq X$, let $\fS[Y] = \bigcup_{y \in Y} \fS[y]$ denote the set of tests that contain at least one element of $Y$. Let $t \in \N_+$ and assume $|X| > t$. Then
	\begin{enumerate}[(a)]
		\itemsep0em
		\item $(X, \fS)$ is called \emph{$t$-separable} if there are no two distinct $Y, Z \subseteq X$ such that $|Y| = |Z| = t$ and $\fS[Y] = \fS[Z]$, and
		\item $(X, \fS)$ is called \emph{$t$-disjunct} if there is no $Y \subseteq X$ and $x \in X \setminus Y$ such that $|Y| = t$ and $\fS[x] \subseteq \fS[Y]$.
	\end{enumerate}

Intuitively, $\fS[Y]$ and $\fS[Z]$ are the collections of tests that are positive if $Y$, respectively $Z$, is the set of defective items. The set system is a valid configuration of tests if and only if it can distinguish the events ``$Y$ is the defective set'' and ``$Z$ is the defective set'' for all distinct size-$t$ sets $Y,Z \subseteq X$, as captured by the first definition. Thus, $t$-separability is \emph{necessary and sufficient} for non-adaptive group testing with $t$ defectives.

Observe that $t$-separability only guarantees that the defective set \emph{can be inferred} from the test results, but does not imply an efficient algorithm to do so. 
A na\"ive approach is to check, for all $\binom{|X|}{t}$ size-$t$ subsets of items, whether they are consistent with the test outcomes. We are not aware of significantly faster algorithms to determine the defective set with an arbitrary $t$-separable set system.

The stronger property of $t$-disjunctness is \emph{sufficient} for an efficient algorithm. Intuitively, $t$-disjunctness guarantees that every non-defective item appears in some test that contains none of the $t$ defectives. We can therefore simply discard all items that appear in at least one negative test, so that only the $t$ defectives remain~\cite[\S\,7.1]{DuHwang1993}. 

In the following, we ignore algorithmic aspects of group testing and focus on studying $t$-separable and $t$-disjunct set systems induced by geometric ranges. We first state the following simple facts:

	\begin{lemma}[Du and Hwang \cite{DuHwang1993}] \label{p:sep-basic-props}
		For each set system $(X, \fS)$ and each $t \ge 1$,
		\begin{align*}
		& (t+1)\text{-separable} \implies t\text{-separable}, \text{ and}\\
		& (t+1)\text{-disjunct} \implies t\text{-disjunct}  \implies t\text{-separable}.
		\end{align*}
	\end{lemma}

		Let $X \subseteq \R^d$ be a finite set of points and let $\fS$ %
be a finite set of geometric shapes in $\R^d$. We call $(X, \fS')$ the set system \emph{induced by $(X, \fS)$}, where $\fS' = \{ R \cap X \mid R \in \fS \}$. %
We use $(X, \fS)$ and $(X, \fS')$ interchangeably, observing that the definitions of $t$-separability and $t$-disjunctness can be equivalently applied to $(X, \fS)$ instead of $(X, \fS')$. %
	Given $d \ge 1$, an \emph{(axis-parallel) $d$-rectangle} (also referred to as box or hyperrectangle) is the cartesian product of $d$ closed intervals. %
We denote the set of all $d$-rectangles in $\R^d$ by $\fR_d$. %

	We wish to determine, given $n$ and $t$, the maximum $m$, such that there exists a $t$-separable (resp.\ $t$-disjunct) set system $(X,\fS)$ with $|X| = m$ and $|\fS| = n$.
Define $\rS^d_t(n)$ to be the maximum $m$ for which there exist a set $X \subseteq \R^d$ of size $m$ and a set $\fS \subseteq \fR_d$ of size $n$ such that $(X, \fS)$ is $t$-separable. In other words, $\rS^d_t(n)$ is the maximum number of points in $\R^d$, among which we can identify $t$ defectives, using $n$ rectangle-tests. Define $\rD_t^d(n)$ accordingly for $t$-disjunctness. Our goal is to understand the asymptotic behaviour of $\rS^d_t(n)$ and $\rD^d_t(n)$. 

An easy observation is that both $\rS^d_t(n)$ and $\rD^d_t(n)$ are monotone in both $t$ and $d$. 

	\begin{proposition}\label{p:monotone}

For all $t,d \in \N_{+}$,
\begin{eqnarray*}
\rS^d_t(n) \geq \rS^d_{t+1}(n),&\mbox{~and~~} & \rS^d_t(n) \leq \rS^{d+1}_t(n),\\
\rD^d_t(n) \geq \rD^d_{t+1}(n),&\mbox{~and~~} & \rD^d_t(n) \leq \rD^{d+1}_t(n).
\end{eqnarray*}
\end{proposition}

Monotonicity in $t$ follows from Lemma~\ref{p:sep-basic-props}, and monotonicity in $d$ follows from the fact that a $d$-dimensional configuration can always be embedded in $\R^{d+1}$.

Recall that in the general, combinatorial case, $n$ tests can handle $2^{\Theta(n)}$ items, assuming that the number of defectives $t$ is constant~\cite[\S\,7]{DuHwang1993}. %
The next upper bound shows that restricting the tests to $d$-rectangles changes the situation considerably, even when $t=1$. %

	\begin{proposition}\label{p:rect-trivial-upper-bound}
		For all $d \in \N_{+}$, ~~~$\rS_1^d(n) \le (2n-1)^d$.

	\end{proposition}
	\begin{proof}
		Let $(X,\fS)$ be a $1$-separable set system induced by $d$-rectangles, where $|\fS| = n$. In each dimension, the boundaries of the $d$-rectangles are defined by at most $2n$ distinct coordinates. These coordinates define a hypergrid with at most $(2n-1)^d$ cells. Assume $(X, \fS)$ to be in general position, so no point is on a cell boundary. Two arbitrary points (items) $x,y \in X$ cannot be in the same cell, as otherwise they would be contained in the same set of $d$-rectangles (tests), and $\fS[Y \cup \{x\}] = \fS[Y \cup \{y\}]$ for all $Y \subseteq X$. Thus, $|X| \le (2n-1)^d$.
	\end{proof}

	For arbitrary $d$, given a set $X$ of $n$ points in $\R^d$, we can choose $n$ disjoint $d$-rectangles so that each rectangle contains a single point of $X$ in its interior. This structure is clearly $t$-disjunct for all $t$. Together with \Cref{p:sep-basic-props} and Propositions~\ref{p:monotone} and \ref{p:rect-trivial-upper-bound}, we obtain:
	\begin{align*}
		n \le \rD_t^d(n) \le \rS_t^d(n) \in \fO(n^d).
	\end{align*}

	\subsection{Main results}
	
	Note that in the one-dimensional case ($d=1$), the trivial upper and lower bounds match. In the case of a single defective item ($t=1$), we show that the upper bound is tight for all $d$. 
	For all other cases, i.e.\ when $t,d \ge 2$, we show that both the trivial lower bound and the trivial upper bound can be improved. %
	
\paragraph{Lower bounds.}	In \Cref{sec:grid-lines}, we analyze a $d$-dimensional grid-line construction. Transforming it to a combinatorially equivalent set system that is induced by two-dimensional rectangles, we obtain the following results:
	\begin{restatable}{theorem}{restateGridLineLowerBounds}
		\label{p:grid-line-lower-bounds}
		$\rS^2_3(n) \in \Omega( n^{\sfrac{3}{2}} )$, and for even $t \ge 6$, $\rS^2_t(n) \in \Omega( n^{1 + \sfrac{2}{t}} )$.
	\end{restatable}
	
	By monotonicity (\Cref{p:monotone}), it follows e.g.\ that $\rS^2_2(n) \in \Omega( n^{\sfrac{3}{2}} )$ and $\rS^2_4(n) \in \Omega( n^{\sfrac{4}{3}} )$.
 Bounds for odd $t$ also follow by monotonicity.
 	
	In \Cref{sec:hyperplanes}, we consider an arrangement of \emph{hyperplanes} on grid points, and again construct equivalent rectangle-induced set systems, obtaining:
	
	\begin{restatable}{theorem}{restateHyperplaneRectLowerBounds}
		\label{p:hyperplane-rect-lower-bounds}
		For constants $d, t \in \N_+$, $\rD_t^d(n) \in \Omega \left( n^{1 + \left\lfloor\frac{d-1}{t}\right\rfloor} \right).$
	\end{restatable}

In the special case $t=1$, i.e.\ with a single defective item, the result matches the upper bound of \Cref{p:rect-trivial-upper-bound}. It follows that $\rD^d_1(n) \in \Theta( n^d )$ and $\rS^d_1(n) \in \Theta( n^d )$ for all $d \in \N_+$, and our constructions are asymptotically optimal for this case. 
	
	Note that \Cref{p:hyperplane-rect-lower-bounds} only implies meaningful lower bounds when $t < d$. In \Cref{sec:long-rectangles}, we present superlinear lower bounds for the other cases.
	
	\begin{restatable}{theorem}{restateLongRectsLowerBounds}
		\label{p:long-rects-lower-bounds}
		For $d, t \in \N_{+}$, with $t \ge d$, $\rD^d_t(n) \in \Omega \left( n^{1 + \frac{1}{2+t-d}} \right)$.
	\end{restatable}
	
In particular, in the two-dimensional case $\rD^2_t(n) \in \Omega{(n^{1+\sfrac{1}{t}})}$.
	
		Finally, in \Cref{sec:grid-subspaces}, we consider a generalization of the grid construction of \Cref{sec:grid-lines} that replaces axis-parallel lines with axis-parallel affine subspaces, yielding a slight improvement over \Cref{p:hyperplane-rect-lower-bounds} if $d$ is a multiple of $t$.
	
	\begin{restatable}{theorem}{restateGridSubspacesLowerBound}
		\label{p:grid-subspaces-lower-bound}
For $d, t \ge 2$, $\rD_t^d(n) \in \Omega\left( n^{\frac{d+1}{t}} \right) = \Omega\left( n^{1 + \frac{d+1-t}{t}} \right) $.
	\end{restatable}
	
	\paragraph{Upper bounds.}
	In \Cref{sec:rect-upper-bounds}, we develop a technique based on avoided patterns (variants and extensions of the K\H{o}v\'ari-S\'os-Tur\'an theorem~\cite{KovariSosEtAl1954}), obtaining the following result for 2-separability:
	
	\begin{restatable}{theorem}{restateDecompSepUpperBoundTwo}\label{p:decomp-sep-upper-bound-two-dim}
		For $d \ge 2$, $\rS_2^d(n) \in \fO \left( n^{d-\sfrac{1}{3} } \right)$.
	\end{restatable}
	
	In particular, in the two-dimensional case $\rS_2^2(n) \in \fO(n^{\sfrac{5}{3}})$. For 3-separability and higher, we obtain the stronger bound:
	
	\begin{restatable}{theorem}{restateDecompSepUpperBoundHigher}
		\label{p:decomp-sep-upper-bound}
		For $d \ge 2$ and $t \ge 3$, $\rS_t^d(n) \in \fO( n^{d - \sfrac{1}{2}})$.
	\end{restatable}
	
	In particular, together with \Cref{p:long-rects-lower-bounds}, we have the asymptotically tight result $\rS_3^2(n) \in \Theta( n^{\sfrac{3}{2}})$. For disjunctness we obtain:

	\begin{restatable}{theorem}{restateDecompDisjUpperBound}
		\label{p:decomp-disj-upper-bound}
		For $d, t \ge 2$, $\rD_t^d(n) \in \fO \left( n^{d-1+\frac{1}{\min(d,t)}} \right)$.
	\end{restatable}
	
	In particular, $\rD_2^2(n) \in \Theta( n^{\sfrac{3}{2}})$. In Appendix~\ref{appc} we summarize the best upper and lower bounds for $\rD_t^d(n)$ and $\rS_t^d(n)$, for small values of $t$ and $d$.\\

A remark is in order, regarding \emph{general position assumptions}, as in the proof of \Cref{p:rect-trivial-upper-bound}. We say that a configuration of points $X$ and $d$-rectangles $\fS$ is \emph{in general position} if, for each $i \in [d]$, the $i$-th coordinates defining the points in $X$ and the rectangles in $\fS$ are all distinct. In other words, each axis-parallel hyperplane contains at most one point and bounds at most one rectangle, not both at once.

In our proofs it is sometimes useful to assume $(X, \fS)$ to be in general position. We argue that this assumption is without loss of generality, as we can always perturb a configuration to obtain a combinatorially equivalent one that is in general position. 

Indeed, let $\varepsilon$ be the smallest positive difference between a pair of coordinates of points in $X$ and/or corners of rectangles in $\fS$. Extend each rectangle by $\varepsilon/3$ in each of the $2d$ axis-parallel directions. By the choice of $\varepsilon$, %
the point-rectangle containment relation is unchanged, and after the transformation
no point is aligned with a rectangle boundary. By further perturbing the coordinates of points and rectangle corners by less than $\varepsilon/6$, we achieve general position.

\subsection{Further related work}
To our knowledge, group testing with orthogonal range queries has not been considered before, although as natural measures of set system complexity, $t$-separability and $t$-disjunctness have been studied on their own. Equivalent concepts include \emph{superimposed codes}~\cite{KautzSingleton1964}, \emph{$t$-cover-free families} (duals of $t$-disjunct set systems) \cite{ErdoesFranklEtAl1985, Bshouty14}, and \emph{separating systems} (1-separable set systems)~\cite{Renyi1965, Katona1966,  AhlswedeWegener1987,LangiNaszodiEtAl2016,GerbnerToth2013}.

We focus on results for geometrically-defined set systems. As they refer to \emph{worst-case} placements of points, the bounds obtained in the referred works are not directly comparable with ours. Gerbner and T\'oth~\cite{GerbnerToth2013} study $1$-separable set systems induced by \emph{convex} sets over points \emph{in general position}. Translated into our setting their results imply that $n$ convex tests can handle $O(n \log{n})$ arbitrary points (with a single defective) if the general position assumption is made, and $O(n)$ points otherwise. 
Boland and Urrutia~\cite{BolandUrrutia} study set systems induced by hyperplanes in a similar setting, obtaining $\Theta(n)$ bounds when the dimension $d$ is constant.
Gledel and Parreau~\cite{GledelParreau2017} consider 
a similar problem where $\fS$ are disks; Harvey et al.~\cite{Harvey} and Cheraghchi et al.~\cite{CheraghchiEtAl2010} consider sets $\fS$ restricted to paths of a given graph.

L\'angi et al.~\cite{LangiNaszodiEtAl2016} study $1$-separation, and two generalizations: \emph{intersection-separation} and \emph{containment-separation}. The first notion corresponds exactly to the $t$-separation studied in our paper, requiring that for all pairs of size-$t$ sets, some test properly intersects one but not the other. The second notion, $t$-containment-separation is similar, with ``intersects'' replaced by ``contains''. (Notice that in the case $t=1$ the two notions are identical.) The bounds obtained by L\'angi et al.\ are linear, up to poly-logarithmic factors. Again, these results are not comparable to ours, as they refer to \emph{worst-case} placements of points in general position, with \emph{convex subsets} of $\mathbb{R}^d$ as tests, whereas we study \emph{best-case} placements of points with \emph{axis-parallel boxes} as tests.

Several other variants of the group testing problem have been considered (see \cite{DuHwang1993, AJS} and references therein). In \emph{probabilistic group testing} it is only required to recover the defective set with high probability. In \emph{noisy group testing} the tests might err with some probability~\cite{AtiaSaligrama2009}, or possibly in an adversarial way, i.e.\ testing with ``liars''~\cite[\S\,5]{DuHwang1993}, \cite{Pelc02}. In \emph{quantitative group testing} the test outcomes indicate the number of defective items in the test~\cite{Alaoui,Alaoui2}; also related are various \emph{coin-weighing} problems, see e.g.\ \cite{Bshouty09}.
	
	\section{Grids of axis-parallel lines}\label{sec:grid-lines}

In this section we consider a specific family of set systems defined by points and lines, and fully characterise its $t$-separability, respectively $t$-disjunctness. We then transform this set system to obtain a combinatorially equivalent set system whose sets are induced by $2$-rectangles.

Let $n \ge 2$, let $P_d = [n]^d$, and define $\fL_d$ to be the set of axis-parallel lines that intersect $P_d$. In other words, $\fL_d$ is the set of grid lines in the $d$-dimensional $n \times n \times \dots \times n$ hypergrid. Observe that $|P_d| = n^d$ and $|\fL_d| = d n^{d-1}$.
	
	We start with the two-dimensional case and argue that $(P_2, \fL_2)$ is $1$-disjunct. Suppose this is not the case. Then, there exist distinct points $x, y \in P$ such that $\fL_2[x] \subseteq \fL_2[y]$. As $x$ is contained in at least one line of $\fL_2$ that does not contain $y$, this is impossible. By \Cref{p:sep-basic-props}, $(P_2, \fL_2)$ is also $1$-separable. 

We argue next that $(P_2, \fL_2)$ is not $2$-separable. Indeed, as shown in \Cref{fig:grid-line}~(a), we can find two distinct sets of points in $P_2$, both of size two, such that both sets intersect the same four grid lines. (The two sets contain opposite corners of a grid cell.) By \Cref{p:sep-basic-props}, $(P_2, \fL_2)$ is not $t$-separable and not $t$-disjunct for all $t \geq 2$. This settles the two-dimensional case. We continue with bounds regarding disjunctness in higher dimensions.
	
	\begin{proposition}\label{p:grid-line-not-disjunct}
		Let $d \ge 2$. Then $(P_d, \fL_d)$ is not $d$-disjunct.
	\end{proposition}
	\begin{proof}
		Consider a point $p \in P_d$ and the set $\fL_d[p]$ of $d$ grid lines containing $p$. Choose one more point on each line in $\fL_d[p]$, yielding a set $Q \subseteq P_d \setminus \{p\}$ of size $d$. Now $\fL_d[p] \subseteq \fL_d[Q]$, so $(P_d,\fL_d)$ is not $d$-disjunct.
	\end{proof}
	
	\begin{proposition}\label{p:grid-line-disjunct}
		Let $d \ge 2$. Then $(P_d, \fL_d)$ is $(d-1)$-disjunct.
	\end{proposition}
	\begin{proof}
		Consider a point $a \in P_d$ and a set $B \subseteq P_d \setminus \{a\}$ such that $\fL_d[a] \subseteq \fL_d[B]$, i.e.\ the points in $B$ hit all lines that intersect $a$. Two lines in $\fL_d[a]$ intersect only in $a$, which means that $|\fL_d[a] \cap \fL_d[b]| \le 1$ for all $b \in B$, implying $|B| \ge |\fL_d[a]| = d > d - 1$.
	\end{proof}

	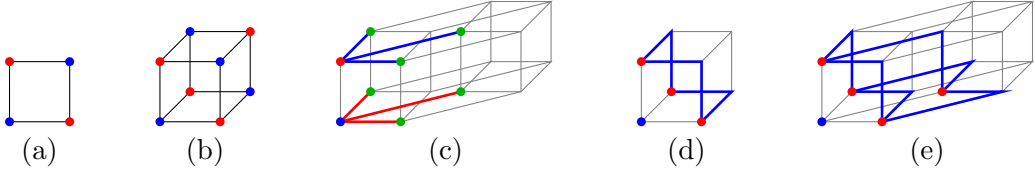
\begin{figure}
		\centering
		\begin{tikzpicture}[
			scale = {0.4},
			point/.style = {fill, circle, inner sep = 1.2pt},
			inA/.style = {blue},
			inB/.style = {red},
			inAB/.style = {green!70!black},%
			line/.style = {},
			gline/.style = {gray},
			aline/.style = {blue, line width=1pt},
			bline/.style = {red, line width=1pt}
		]
			\newcommand{\pa}[1]{\node[point, inA] at (#1) {};}
			\newcommand{\pb}[1]{\node[point, inB] at (#1) {};}
			\newcommand{\pab}[1]{\node[point, inAB] at (#1) {};}
			
			\begin{scope}
				\draw[line] (0,0) -- (0,2) -- (2,2) -- (2,0) -- (0,0);
				\pa{0,0}
				\pb{0,2}
				\pa{2,2}
				\pb{2,0}
				
				\node at (1, -1) {(a)};
			\end{scope}
			
			\begin{scope}[shift={(5,0)}]
				\draw[line] (0,0) -- (0,2) -- (2,2) -- (2,0) -- (0,0);
				\draw[line] (1,1) -- (1,3) -- (3,3) -- (3,1) -- (1,1);
				\draw[line] (0,0) -- (1,1);
				\draw[line] (0,2) -- (1,3);
				\draw[line] (2,0) -- (3,1);
				\draw[line] (2,2) -- (3,3);
				\pa{0,0}
				\pb{0,2}
				\pa{2,2}
				\pb{2,0}
				\pb{1,1}
				\pa{1,3}
				\pb{3,3}
				\pa{3,1}
				
				\node at (1.5, -1) {(b)};
			\end{scope}
			
			\begin{scope}[shift={(11,0)}]
				\draw[gline] (0,0) -- (0,2) -- (2,2) -- (2,0) -- (0,0);
				\draw[gline] (1,1) -- (1,3) -- (3,3) -- (3,1) -- (1,1);
				\draw[gline] (5,2) -- (5,4) -- (7,4) -- (7,2) -- (5,2);
				\draw[gline] (4,1) -- (4,3) -- (6,3) -- (6,1) -- (4,1);
				\draw[gline] (0,0) -- (1,1) -- (5,2) -- (4,1) -- (0,0);
				\draw[gline] (0,2) -- (1,3) -- (5,4) -- (4,3) -- (0,2);
				\draw[gline] (2,0) -- (3,1) -- (7,2) -- (6,1) -- (2,0);
				\draw[gline] (2,2) -- (3,3) -- (7,4) -- (6,3) -- (2,2);
				
				\draw[bline] (0,0) -- (1,1);
				\draw[bline] (0,0) -- (2,0);
				\draw[bline] (0,0) -- (4,1);
				
				\draw[aline] (0,2) -- (1,3);
				\draw[aline] (0,2) -- (2,2);
				\draw[aline] (0,2) -- (4,3);
				
				\pa{0,0}
				\pb{0,2}
				\pab{1,1}
				\pab{1,3}
				\pab{2,0}
				\pab{2,2}
				\pab{4,1}
				\pab{4,3}
				
				\node at (3.5,-1) {(c)};
			\end{scope}
			
			\begin{scope}[shift={(21,0)}]
				\draw[gline] (0,0) -- (0,2) -- (2,2) -- (2,0) -- (0,0);
				\draw[gline] (1,1) -- (1,3) -- (3,3) -- (3,1) -- (1,1);
				\draw[gline] (0,0) -- (1,1);
				\draw[gline] (0,2) -- (1,3);
				\draw[gline] (2,0) -- (3,1);
				\draw[gline] (2,2) -- (3,3);
				
				\draw[aline] (0,2) -- (2,2) -- (2,0) -- (3,1) -- (1,1) -- (1,3) -- (0,2);
				\pa{0,0}
				\pb{0,2}
				\pb{1,1}
				\pb{2,0}
				
				\node at (1.5, -1) {(d)};
			\end{scope}
			
			\begin{scope}[shift={(27,0)}]
				\draw[gline] (0,0) -- (0,2) -- (2,2) -- (2,0) -- (0,0);
				\draw[gline] (1,1) -- (1,3) -- (3,3) -- (3,1) -- (1,1);
				\draw[gline] (5,2) -- (5,4) -- (7,4) -- (7,2) -- (5,2);
				\draw[gline] (4,1) -- (4,3) -- (6,3) -- (6,1) -- (4,1);
				\draw[gline] (0,0) -- (1,1) -- (5,2) -- (4,1) -- (0,0);
				\draw[gline] (0,2) -- (1,3) -- (5,4) -- (4,3) -- (0,2);
				\draw[gline] (2,0) -- (3,1) -- (7,2) -- (6,1) -- (2,0);
				\draw[gline] (2,2) -- (3,3) -- (7,4) -- (6,3) -- (2,2);
				
				\draw[aline] (0,2) -- (1,3) -- (1,1);
				\draw[aline] (1,1) -- (3,1) -- (2,0);
				\draw[aline] (2,0) -- (2,2) -- (0,2);
				\draw[aline] (4,1) -- (5,2) -- (1,1);
				\draw[aline] (4,1) -- (6,1) -- (2,0);
				\draw[aline] (4,1) -- (4,3) -- (0,2);
				
				\pa{0,0}
				\pb{2,0}
				\pb{1,1}
				\pb{0,2}
				\pb{4,1}
				
				\node at (3.5,-1) {(e)};
			\end{scope}
		\end{tikzpicture}
		\caption{\emph{(a)} $(P_2, \fL_2)$ is not $2$-separable; red and blue points indicate sets $A,B \subseteq P_2$, where $\fL_2[A] = \fL_2[B]$. 
			\emph{(b)} $(P_3,\fL_3)$ is not $4$-separable.
			\emph{(c)} $(P_d,\fL_d)$ is not $(2d-1)$-separable (example with $d=4$); red, blue, and green points indicate points in $A \setminus B$, $B \setminus A$, resp.\ $A \cap B$.
			\emph{(d), (e)} First step in the proof of \Cref{p:grid-line-sep} for $d \in \{3,4\}$. The blue point is $a_0$ and the set of red points is $B'$. The blue lines $\fL'$ have to be hit by $A$.}
		\label{fig:grid-line}
	\end{figure}
	
	Determining $t$-separability of $(P_d,\fL_d)$ is somewhat more complicated. Here, $d = 3$ is also a special case. %
	
	\begin{proposition}
		$(P_3, \fL_3)$ is not $4$-separable.
	\end{proposition}
	\begin{proof}
		Consider a set of $8$ points in $P_3$ forming the corners of a $3$-rectangle (i.e.\ cube). Split these points into two sets $Y$ and $Z$, both of size $4$, as indicated by the coloring in \Cref{fig:grid-line}~(b). (The sets correspond to the odd, resp.\ even layers of the cube.) Then, $\fL_3[Y] = \fL_3[Z]$, as both $Y$ and $Z$ hit exactly the lines supporting the edges of the cube. Thus, $(P_3, \fL_3)$ is not $4$-separable.
	\end{proof}
	
	The constructions in \Cref{fig:grid-line}~(a) and (b) can be easily generalized to higher dimensions, showing that $(P_d, \fL_d)$ is not $2^{d-1}$-separable. The correct bounds are, however, much lower, as we show next. 
	
	\begin{proposition}
		Let $d \ge 4$. Then $(P_d, \fL_d)$ is not $(2d-1)$-separable.
	\end{proposition}
	\begin{proof}
		We use the following construction, shown in \Cref{fig:grid-line}~(c). Let $x, y \in P_d$ be distinct points contained in the same grid line $L$. Both $x$ and $y$ are contained in $d-1$ lines apart from $L$, and these $2d-2$ lines are pairwise distinct. Consider an additional point on each of these lines, and call the resulting set of $2d-2$ points $Z$. Now we have $\fL_d[Z \cup \{x\}] = \fL_d[Z \cup \{y\}]$ as the two sets of points hit the same set of lines.
	\end{proof}
	
	\begin{proposition}\label{p:grid-line-sep}
		$(P_3,\fL_3)$ is 3-separable, and $(P_d, \fL_d)$ is $(2d-2)$-separable for $d \ge 4$.
	\end{proposition}
	
	We first fix some notation and make some observations. For $i \in [d]$ and $x \in P_d$, let $L_i(x) \in \fL_d$ denote the line that contains $x$ and all points that differ from $x$ only in the $i$-th coordinate. In particular, $\fL_d[x] = \{L_i(x) \mid i \in [d]\}$.
	
	Let $A, B \subseteq P_d$ be distinct sets of points of the same cardinality that intersect the same set of lines, i.e.\ $\fL_d[A] = \fL_d[B]$. We have to show that the cardinalities of $A$ and $B$ are at least $2d-1$ (respectively $4$, if $d=3$). 

Let $a_0 \in A \setminus B$. Each line in $\fL_d[a_0]$ must intersect some point in $B$. %
For $i \in [d]$, let $b_i$ denote a point in $L_i(a_0) \cap B$. Note that these points are distinct and let $B' = \{b_1, b_2, \dots, b_d\}$. Consider now the set of lines $\fL' = \{ L_i(b_j) \mid i,j \in [d], i \neq j \}$. Notice that these lines contain some point in $B'$, but do not contain $a_0$. See \Cref{fig:grid-line}~(d) and (e) for illustration. Further observe that
	\begin{enumerate}[(i)]
		\itemsep0em
		\item The set of lines $L_1(b_j), L_2(b_j), \dots, L_d(b_j)$ intersect in $b_j$;\label{item:grid-line-bar-sep:b_j-lines}
		\item the lines $L_i(b_j)$ and $L_j(b_i)$ intersect in a unique point for all $i \neq j$; and\label{item:grid-line-bar-sep:L_i-b_j-lines}
		\item no other pair of lines in $\fL'$ intersects.\label{item:grid-line-bar-sep:no-more-lines}
	\end{enumerate}
	
Claim (\ref{item:grid-line-bar-sep:b_j-lines}) is immediate from the definition of $L_i(b_j)$. Claim (\ref{item:grid-line-bar-sep:L_i-b_j-lines}) follows, as $L_i(b_j) \cap L_j(b_i)$ is the unique point that has the $i$-coordinate of $b_i$, the $j$-coordinate of $b_j$, and agrees with $a_0$ in all other coordinates. Claim (\ref{item:grid-line-bar-sep:no-more-lines}) follows, as pairs of lines $L_i(b_j)$ and $L_k(b_\ell)$ are parallel if $i=k$, and not coplanar otherwise. %
	
	Recall that each line $L \in \fL'$ must be hit by some point $a \in A \setminus \{a_0\}$. If $d = 3$, then the lines of $\fL'$ form a $6$-cycle, as shown in \Cref{fig:grid-line} (d). Clearly, to hit these lines, we need at least $3$ points, so $|A \setminus \{a_0\}| \ge 3$, and thus $|A| \ge 4$. This settles the case $d = 3$.

For $d \ge 4$, we argue as follows. %
From the above claims it follows that points of $a \in A \setminus \{a_0\}$ can be of three types: (i) points of $B'$, (ii) intersection points $L_i(b_j) \cap L_j(b_i)$ for $i \neq j$, and (iii) other points.

Observe that points of type (i) hit at most $d-1$ lines of $\fL'$, points of type (ii) hit at most two lines of $\fL'$, and points of type (iii) hit at most one line of $\fL'$.
We claim that if $A$ is sufficiently small, then it must contain all possible points of type (i).

	\begin{lemma} \label{p:subd-K_d}
		If $|A| \le 2d - 2$, then $B' \subseteq A$.
	\end{lemma}
	\begin{proof}

Let $k_1$, $k_2$, and $k_3$ denote, respectively the number of points of type (i), (ii), and (iii) in $A \setminus \{a_0\}$. Note that $k_1 \leq d$.

As $|\fL'| = (d-1)d$, to hit all lines, we must have $(d-1)k_1 + 2k_2 + k_3 \geq d(d-1)$.

Suppose $k_1 \leq d-3$. As $|A \setminus \{a_0\}| = k_1 + k_2 + k_3 \leq 2d-3$, we have $(d-1)k_1 + 2k_2 + k_3 \leq (d-1)k_1 + 2(2d-3-k_1) < d(d-1)$, a contradiction, as $d \ge 4$.

Suppose $k_1 = d-1$. Then, there is a point in $B' \setminus A$. Let this point be $b_j$. All $d-1$ lines $L_i(b_j)$ for $i \neq j$ must be hit by some point of type (ii) or (iii), moreover, no point can hit more than one of these lines, so $k_2+k_3 \geq d-1$ must hold. However, $k_2 + k_3 \leq 2d-3 - (d-1) = d-2$, a contradiction.   	

Suppose $k_1 = d-2$. Then there are two points in $B' \setminus A$. Let these points be $b_j$ and $b_k$. Consider the $2(d-2)$ lines $L_i(b_j)$ and $L_i(b_k)$ for $i \neq j$ and $i \neq k$. All these lines must be hit by some point of type (ii) or (iii), moreover, no point can hit more than one of these lines, so $k_2 + k_3 \geq 2(d-2)$ must hold. However, $k_2 + k_3 \leq 2d-3 - (d-2) = d-1$, a contradiction, as $d \ge 4$.

The only remaining possibility is $k_1 = d$, so $A$ must contain all $d$ points in $B'$.
	\end{proof}

It remains to show that $|A| \geq 2d - 1$. Suppose for contradiction that $|A| \le 2d - 2$. By \Cref{p:subd-K_d}, $B' \subseteq A$. 

Let $b_0 \subseteq B \setminus A$ (recall that $A \neq B$). By symmetry, for each $i \in [d]$, there is a point $a_i \in A$ that is contained in $L_i(b_0)$. Let $A' = \{a_1, a_2, \dots, a_d\}$. Again, applying \Cref{p:subd-K_d}, we have $A' \subseteq B$. 

Observe that, as $b_0 \notin A$ (by definition), we have $b_0 \notin A' \cup B'$ (Since $A',B' \subseteq A$). Thus, $A' \cup B' \subseteq B \setminus \{b_0\}$. As $|B| \leq 2d-2$, it follows that $|A' \cup B'| \leq 2d-3$ must hold. Recall that $|A'| = |B'| = d$, and therefore $|A' \cap B'| \geq 3$.

	Let $\{x_1,x_2,x_3\} \subseteq A' \cap B'$. %
	Observe that $x_1$ differs from $a_0$ in exactly one coordinate (say, coordinate $i$) and from $b_0$ in exactly one coordinate (say, coordinate $j$). Thus, as $a_0 \neq b_0$, $a_0$ and $b_0$ must differ precisely in coordinates $i$ and $j$. Now $x_2$ and $x_3$, to be different from $x_1$, both must differ only in $j$ from $a_0$ and in $i$ from $b_0$. This implies $x_2 = x_3$, a contradiction, concluding the proof of Proposition~\ref{p:grid-line-sep}.
	
	\subsection{Mapping into two dimensions}
	
	Note that, in terms of point inclusion, an axis-parallel line is equivalent to a (sufficiently long and thin) hyperrectangle. Recall that $|P_d| = n^d$ and $|\fL_d| = d n^{d-1}$, so we already have $\rD^d_{d-1}(n) \in \Omega( n^{1 + 1/(d-1)})$ from \Cref{p:grid-line-disjunct} and $\rS^3_3(n) \in \Omega(n^{\sfrac{3}{2}})$, $\rS^d_{2d-2}(n) \in~\Omega(n^{1 + 1/(d-1)})$ for $d \ge 4$ from \Cref{p:grid-line-sep}. 
	
	The first bound will be improved in subsequent sections. We now strengthen the second and third bounds by %
	showing that $(P_d, \fL_d)$ is isomorphic to a \emph{two-dimensional} rectangle-induced set system. This proves:
	\restateGridLineLowerBounds*
	
	We construct the necessary isomorphism directly, that is, we define functions $f \colon P_d \rightarrow~\R^2$ and $g \colon \fL_d \rightarrow \fR_2$, such that, for each $x \in P_d$ and $L \in \fL_d$, we have $x \in L$ if and only if $f(x) \in g(L)$.
	
	For a set $X \subseteq P_d$, we write $f(X) = \{ f(x) \mid x \in X \}$. Assuming that $f$ is injective, the above condition can be written as $f(P_d \cap L) = f(P_d) \cap g(L)$. We prove the two directions separately (for all $L \in \fL_d$):
	\begin{align*}
		& f(P_d \cap L) \subseteq f(P_d) \cap g(L), & \text{(i)} \\
		& f(P_d) \cap g(L) \subseteq f(P_d \cap L). & \text{(ii)}
	\end{align*}
	
	We start with the definition of $f$ and some observations. Let $\varepsilon = \frac{1}{n+1}$, and let
	\begin{align*}
		f( x_1, x_2, \dots, x_d ) = \left( \sum_{i=1}^d \varepsilon^{i-1} x_i, \sum_{i=1}^d \varepsilon^{d-i} x_i \right).
	\end{align*}
	
	Intuitively, as $i$ grows, $x_i$ contributes less to $f(x)_1$ and more to $f(x)_2$. Moreover, it is easy to show that $\sum_{j=i}^d \varepsilon^{j-1} n < \varepsilon^{i-2}$ for all $i \in [d]$. This means that changing $x_i$ by only one has a larger effect on $f(x)_1$ then all possible changes to the variables $x_j$ with $j > i$ combined (recall that $x_j \in [n]$). In particular, for $x, x' \in P_d$, we have $f(x)_1 < f(x')_1$ if and only if $x$ precedes $x'$ in lexicographic order, i.e.\ there is some $i$ such that $x_i < x'_i$ and $x_j = x'_j$ for all $j < i$. Similarly, $f(x)_2 < f(x')_2$ if and only if \emph{the reverse} of $x$ precedes the reverse of $x'$ in lexicographic order, i.e.\ there is some $i$ such that $x_i < x'_i$ and $x_j = x'_j$ for all $j > i$. This implies that $f$ is injective.
	
	We proceed with the definition of $g$. To satisfy (i), $g(L)$ must contain all points in $f(\fP_d \cap L)$. When this is true, increasing the size of $g(L)$ clearly does not help in satisfying (ii). Consequently, we let $g(L)$ be the inclusion-wise minimal axis-parallel rectangle containing all points in $f(P_d \cap L)$. Note that we consider a point or a line segment to be a degenerate rectangle with width and/or height 0.
	
	Now (i) holds by definition. The following lemma implies (ii).
	
	\begin{lemma}\label{p:proj-enclosing-rectangles}
		Let $x \in P_d$ and $L \in \fL_d$ such that $x \notin L$. Then $f(x) \notin g(L)$.
	\end{lemma}
	\begin{proof}
		
		We start with an observation. As an axis-parallel rectangle, $g(L)$ is the product of two intervals, say $g(L) = U \times V$. In particular, $U$ is the smallest interval containing all first coordinates of points in $f(P_d \cap L)$, and $V$ is the smallest interval containing all second coordinates of points in $f(P_d \cap L)$. Thus, using our observation above, a point $y \in P_d$ satisfies $f(y)_1 \in U$ if and only if $y$ is lexicographically between two points in $P_d \cap L$. Similarly, $f(y)_2 \in V$ if and only if the reverse of $y$ is lexicographically between the reverses of two points in $P_d \cap L$.
		
		We now prove the lemma. Let $i$ be the one coordinate that is not fixed by $L$, i.e.\ $L = \{ y \mid \forall j \in [d] \setminus \{i\} : y_i = a_i \} \in \fL_d$ for some $i \in [d]$ and $a_j \in [n]$. As $x \notin L$, there must be some $j$ such that $x_j \neq a_j$. First, suppose that $j < i$. If $x_j < a_j$ then $x$ is lexicographically smaller than all points in $P_d \cap L$. Otherwise, $x$ is lexicographically greater than all points in $P_d \cap L$. In both cases, $f(x)_1 \notin U$, so $f(x) \notin g(L)$.
		
		If $j > i$, then the symmetric argument shows $f(x)_2 \notin V$. This concludes the proof.
	\end{proof}
	
	Having shown that $f$ and $g$ satisfy (i) and (ii), this concludes the proof of \Cref{p:grid-line-lower-bounds}.

	\section{Hyperplanes}\label{sec:hyperplanes}
	
		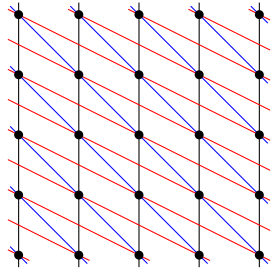
\begin{figure}
		\centering
		\begin{tikzpicture}[
			point/.style={circle, fill, inner sep=1.2pt},
			lineseg0/.style={},
			lineseg1/.style={blue},
			lineseg2/.style={red},
			scale = 0.8
		]
			\draw[lineseg0](0.0, -0.2) -- (0.0,4.2);
\draw[lineseg0](1.0, -0.2) -- (1.0,4.2);
\draw[lineseg0](2.0, -0.2) -- (2.0,4.2);
\draw[lineseg0](3.0, -0.2) -- (3.0,4.2);
\draw[lineseg0](4.0, -0.2) -- (4.0,4.2);
\draw[lineseg1](0.1414213562373095, -0.1414213562373095) -- (-0.1414213562373095,0.1414213562373095);
\draw[lineseg1](1.1414213562373094, -0.1414213562373095) -- (-0.1414213562373095,1.1414213562373094);
\draw[lineseg1](2.1414213562373097, -0.1414213562373095) -- (-0.1414213562373095,2.1414213562373097);
\draw[lineseg1](3.1414213562373097, -0.1414213562373095) -- (-0.1414213562373095,3.1414213562373097);
\draw[lineseg1](4.141421356237309, -0.1414213562373095) -- (-0.1414213562373095,4.141421356237309);
\draw[lineseg1](4.141421356237309, 0.8585786437626906) -- (0.8585786437626906,4.141421356237309);
\draw[lineseg1](4.141421356237309, 1.8585786437626906) -- (1.8585786437626906,4.141421356237309);
\draw[lineseg1](4.141421356237309, 2.8585786437626903) -- (2.8585786437626903,4.141421356237309);
\draw[lineseg1](4.141421356237309, 3.8585786437626903) -- (3.8585786437626903,4.141421356237309);
\draw[lineseg2](0.17888543819998318, -0.08944271909999159) -- (-0.17888543819998318,0.08944271909999159);
\draw[lineseg2](1.1788854381999831, -0.08944271909999159) -- (-0.17888543819998318,0.5894427190999916);
\draw[lineseg2](2.178885438199983, -0.08944271909999159) -- (-0.17888543819998318,1.0894427190999916);
\draw[lineseg2](3.178885438199983, -0.08944271909999159) -- (-0.17888543819998318,1.5894427190999916);
\draw[lineseg2](4.178885438199983, -0.08944271909999159) -- (-0.17888543819998318,2.0894427190999916);
\draw[lineseg2](4.178885438199983, 0.4105572809000084) -- (-0.17888543819998318,2.5894427190999916);
\draw[lineseg2](4.178885438199983, 0.9105572809000084) -- (-0.17888543819998318,3.0894427190999916);
\draw[lineseg2](4.178885438199983, 1.4105572809000084) -- (-0.17888543819998318,3.5894427190999916);
\draw[lineseg2](4.178885438199983, 1.9105572809000084) -- (-0.17888543819998318,4.089442719099992);
\draw[lineseg2](4.178885438199983, 2.4105572809000084) -- (0.8211145618000169,4.089442719099992);
\draw[lineseg2](4.178885438199983, 2.9105572809000084) -- (1.8211145618000169,4.089442719099992);
\draw[lineseg2](4.178885438199983, 3.4105572809000084) -- (2.821114561800017,4.089442719099992);
\draw[lineseg2](4.178885438199983, 3.9105572809000084) -- (3.821114561800017,4.089442719099992);
\node[point] at (0, 0) {};
\node[point] at (0, 1) {};
\node[point] at (0, 2) {};
\node[point] at (0, 3) {};
\node[point] at (0, 4) {};
\node[point] at (1, 0) {};
\node[point] at (1, 1) {};
\node[point] at (1, 2) {};
\node[point] at (1, 3) {};
\node[point] at (1, 4) {};
\node[point] at (2, 0) {};
\node[point] at (2, 1) {};
\node[point] at (2, 2) {};
\node[point] at (2, 3) {};
\node[point] at (2, 4) {};
\node[point] at (3, 0) {};
\node[point] at (3, 1) {};
\node[point] at (3, 2) {};
\node[point] at (3, 3) {};
\node[point] at (3, 4) {};
\node[point] at (4, 0) {};
\node[point] at (4, 1) {};
\node[point] at (4, 2) {};
\node[point] at (4, 3) {};
\node[point] at (4, 4) {};
		\end{tikzpicture}
		
		\caption{The line arrangement of \Cref{p:hyperplanes_on_grid} with $m = 5$, $k = 2$, and $t = 2$. We have $\ell = 3$ and $c_1 = (1,0), c_2 = (1,1), c_3 = (1,2)$.}
	\end{figure}
	
	In this section we consider a configuration of $d$-dimensional grid points with a set of arbitrary, not necessarily axis-parallel, \emph{hyperplanes}. We then construct %
	equivalent set systems induced by hyperrectangles. 
	
	For $k \in \N_{+}$, consider the vectors $c_i = (1, (i-1), (i-1)^2, \dots,(i-1)^{k-1})$ in $\R^k$, for $i \in [\ell]$, where $\ell = (k-1)t + 1$. Observe that every set of $k$ distinct vectors $c_i$ is linearly independent.
	Define the hyperplanes $H_{i,j} = \{ x \in \R^k \mid c_i \cdot x = j \}$ for all $i\in [\ell]$ and $j \in \N_{+}$.
Observe that $H_{i,j}$ and $H_{i',j'}$ are distinct, unless $i=i'$ and $j=j'$, and that $H_{i,j}$ and $H_{i,j'}$ are parallel for all $i,j,j'$. %

	\begin{lemma} \label{p:hyperplanes_on_grid}
			Let $m \in \N_{+}$ with $m > k$, and let $\fS = \{ H_{i,j} \mid i \in [\ell], j \in [\ell^k m] \}$.
		 The set system $([m]^k, \fS)$ is $t$-disjunct.
	\end{lemma}
	\begin{proof}

		Let $x, y \in [m]^k$ be distinct points. We claim that there are at most $(k-1)$ hyperplanes in $\fS$ that contain both $x$ and $y$. Indeed, take $k$ hyperplanes in $\fS$ and consider their intersection. If two of them are parallel, their intersection is empty. Otherwise, their intersections are solutions of an equation system $Ax = b$, where $b$ is some vector and $A \in \N^{k \times k}$ is a matrix whose rows are distinct vectors $c_i$. This means that $A$ has full rank, so if $Ax = b$ has a solution, it is unique. This proves the claim.
		
		Now suppose that $([m]^k, \fS)$ is not $t$-disjunct. Then there is an $x \in [m]^k$ and a $Y \subseteq [m]^k$ with $|Y| = t$ and $x \notin Y$ such that $\fS[x] \subseteq \fS[Y]$. We will reach a contradiction by showing that 
		$|\fS[x]| > |\fS[x] \cap \fS[Y]|$.
		
		First, observe that each $x \in [m]^k$ is contained in exactly $\ell$ hyperplanes, so $|\fS[x]| = \ell > t(k-1)$.
		On the other hand, by our claim,
		\begin{align*}
			& |\fS[x] \cap \fS[Y]| \leq
			 \sum_{y \in Y} |\fS[x] \cap \fS[y]| \leq t(k-1),
		\end{align*}
		yielding the contradiction.
	\end{proof}

	We remark that \Cref{p:hyperplanes_on_grid} implies that for each fixed $t \ge 1$, there is a set $P \subseteq \R^d$ of points and a set $\fS$ of hyperplanes in $\R^d$ such that $(P, \fS)$ is $t$-disjunct and $|P| \in \Omega( |\fS|^d )$. As we show in \Cref{sec:rect-upper-bounds}, such a result is not possible for hyperrectangles.
	
	We now derive the lower bound on $\rD_t^d(n)$. Observe that in \Cref{p:hyperplanes_on_grid} we construct $\ell$ $(\ell^k m)$-partitions of the grid $[m]^k$, each consisting of parallel hyperplanes. The following lemma shows that with $d$-dimensional hyperrectangles, we can construct a set of $d$ arbitrary partitions. The idea is simply to use values of the $i$-th coordinate to encode the $i$-th partition. Each part of a partition corresponds to a $d$-rectangle.
	
	\begin{lemma} \label{p:rect_part_union}
		Let $\Pi_1, \Pi_2, \dots, \Pi_d$ a set of $q$-partitions of some set $X$, and let $\fS = \Pi_1 \cup \Pi_2 \cup \dots \cup \Pi_d \subseteq 2^X$. Then there are mappings $r \colon \fS \rightarrow \fR_d$ and $p \colon X \rightarrow \R^d$ such that for all $x \in X$ and $S \in \fS$, we have $x \in S \Leftrightarrow p(x) \in r(S)$.
	\end{lemma}
	\begin{proof}
		Let $\Pi_i = \{ S_{i,1}, S_{i,2}, \dots, S_{i,q} \}$ for each $i \in [d]$. Define $p$ such that $p(x)_i = j$ if $x \in S_{i,j}$ for $i \in [d], j \in [q]$. Observe that $p(x)_i$ is well-defined, as $\Pi_i$ is a partition. Further, define $R( S_{i,j} ) = \R^{i-1} \times \{j\} \times \R^{d-i}$. Now $x \in S_{i,j}$ if and only if $p(x) \in R(S_{i,j})$. Note that while $R(S_{i,j})$ is technically an axis-parallel hyperplane, it can be replaced by a sufficiently large (and thin) $d$-rectangle.
	\end{proof}
	
	From \Cref{p:hyperplanes_on_grid} (setting $\ell = d$) and \Cref{p:rect_part_union} it follows that for each $k$, $t$, $m$ and $d = (k-1)t+1$ there is a set $X \in \R^d$ of $m^k$ points and a set $\fS$ of $d^{k+1}m$ hyperrectangles, such that $(X, \fS)$ is $t$-disjunct. This implies $\rD_t^d(n) \ge (n / d^{k+1})^k$. After solving for $k$, we obtain:
	
	\restateHyperplaneRectLowerBounds*
	
	\section{Long rectangles}\label{sec:long-rectangles}
	
	We start by verbally describing the technique presented in this section in two dimensions. The main idea is to transform a $t$-disjunct arrangement into a $(t+1)$-disjunct arrangement, so that the ``ratio'' of rectangles to points increases only moderately. Repeating this step yields superlinear lower bounds for $\rD_t^d(n)$ for all $t \ge d$.
	
	An example for a single step in two dimensions is shown in \Cref{fig:long_rect_step}. Start with an arbitrary $t$-disjunct arrangement of $m$ points and $n$ rectangles. If multiple points are on a single horizontal line, perturb them without changing the combinatorial structure. Then make $k$ copies of the arrangement ($k$ is to be optimized later) and place them side by side, horizontally. Finally, for each point $p$ in the original arrangement, add one long and thin rectangle that covers all copies of $p$. This yields an arrangement with $m' = km$ points and $n' = kn+m$ rectangles. In the following, we show that the new arrangement is $(t+1)$-disjunct, and discuss the choice of $k$. We remark that if we start with a two-dimensional arrangement isomorphic to the grid-line arrangement $(P_d, \fL_d)$ from \Cref{sec:grid-lines}, we obtain in one step an arrangement isomorphic to $(P_{d+1}, \fL_{d+1})$.
	
	\begin{figure}[htbp]
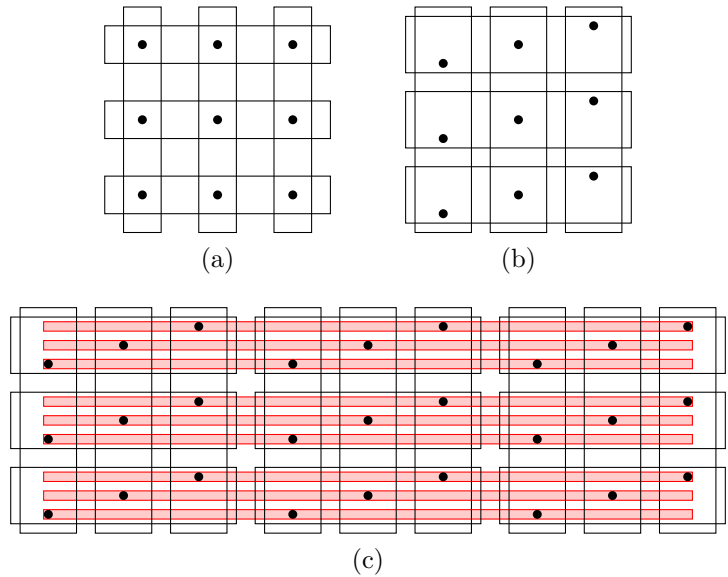

		\centering \small
		\begin{tikzpicture}[
			point/.style={circle, fill, inner sep=1.2pt},
			rect/.style={},
			long-rect/.style={rect, red, fill={red!20!white}},
			scale = 0.25
		]
		\begin{scope}[shift={(5,0)}]
			\input{fig_grid_2d.tex}
			\node at (6,-1.5) {(a)};
		\end{scope}
		
		\begin{scope}[shift={(21,0)}]
			\input{fig_grid_2d_pert.tex}
			\node at (6,-1.5) {(b)};
		\end{scope}
		
		\begin{scope}[shift={(0,-16)}]
			\draw[long-rect] (1.75, 0.75) rectangle (36.25,1.25);
\draw[long-rect] (1.75, 4.75) rectangle (36.25,5.25);
\draw[long-rect] (1.75, 8.75) rectangle (36.25,9.25);
\draw[long-rect] (1.75, 1.75) rectangle (36.25,2.25);
\draw[long-rect] (1.75, 5.75) rectangle (36.25,6.25);
\draw[long-rect] (1.75, 9.75) rectangle (36.25,10.25);
\draw[long-rect] (1.75, 2.75) rectangle (36.25,3.25);
\draw[long-rect] (1.75, 6.75) rectangle (36.25,7.25);
\draw[long-rect] (1.75, 10.75) rectangle (36.25,11.25);
\node[point] at (2, 1.0) {};
\node[point] at (2, 5.0) {};
\node[point] at (2, 9.0) {};
\node[point] at (6, 2.0) {};
\node[point] at (6, 6.0) {};
\node[point] at (6, 10.0) {};
\node[point] at (10, 3.0) {};
\node[point] at (10, 7.0) {};
\node[point] at (10, 11.0) {};
\draw[rect] (0, 0.5) rectangle (12,3.5);
\draw[rect] (0.5, 0) rectangle (3.5,12);
\draw[rect] (0, 4.5) rectangle (12,7.5);
\draw[rect] (4.5, 0) rectangle (7.5,12);
\draw[rect] (0, 8.5) rectangle (12,11.5);
\draw[rect] (8.5, 0) rectangle (11.5,12);
\node[point] at (15, 1.0) {};
\node[point] at (15, 5.0) {};
\node[point] at (15, 9.0) {};
\node[point] at (19, 2.0) {};
\node[point] at (19, 6.0) {};
\node[point] at (19, 10.0) {};
\node[point] at (23, 3.0) {};
\node[point] at (23, 7.0) {};
\node[point] at (23, 11.0) {};
\draw[rect] (13, 0.5) rectangle (25,3.5);
\draw[rect] (13.5, 0) rectangle (16.5,12);
\draw[rect] (13, 4.5) rectangle (25,7.5);
\draw[rect] (17.5, 0) rectangle (20.5,12);
\draw[rect] (13, 8.5) rectangle (25,11.5);
\draw[rect] (21.5, 0) rectangle (24.5,12);
\node[point] at (28, 1.0) {};
\node[point] at (28, 5.0) {};
\node[point] at (28, 9.0) {};
\node[point] at (32, 2.0) {};
\node[point] at (32, 6.0) {};
\node[point] at (32, 10.0) {};
\node[point] at (36, 3.0) {};
\node[point] at (36, 7.0) {};
\node[point] at (36, 11.0) {};
\draw[rect] (26, 0.5) rectangle (38,3.5);
\draw[rect] (26.5, 0) rectangle (29.5,12);
\draw[rect] (26, 4.5) rectangle (38,7.5);
\draw[rect] (30.5, 0) rectangle (33.5,12);
\draw[rect] (26, 8.5) rectangle (38,11.5);
\draw[rect] (34.5, 0) rectangle (37.5,12);
			\node at (19,-1.5) {(c)};
			\end{scope}
		\end{tikzpicture}
		\small 
		\caption{The first step of the long rectangle construction in two dimensions.
			\emph{(a)} Base case for $k=3$, with $2k$ rectangles and $k^2$ points (simplified from what \Cref{p:hyperplanes_on_grid} yields).
			\emph{(b)} Perturbation so that all points have distinct $y$-coordinates.
			\emph{(c)} Three copies of the perturbed configuration, arranged along $x$-coordinate with added long rectangles in red.}
		\label{fig:long_rect_step}
	\end{figure}
	
	\begin{lemma} \label{p:long_rects_constr}
		Let $X \in \R^d$ be a set of $m$ points and let $\fS$ be a set of $n$ $d$-rectangles such that $(X, \fS)$ is $t$-disjunct, and let $k \in \N_+$. Then there is a set $\fS'$ of $d$-rectangles and a set $X' \in \R^d$ of points such that $(X', \fS')$ is $(t+1)$-disjunct, $|\fS'| = k \cdot n + m$ and $|X'| = k \cdot m$.
	\end{lemma}
	\begin{proof}
		Without loss of generality, assume that $X$ is in general position (i.e.\ no two points agree in a coordinate) and that the coordinates of all points in $X$ and corners of hyperrectangles in $\fS$ are integers in $[m]$.\footnote{We can assume $(X, \fS)$ to be in general position. Then, replacing each rectangle by the smallest combinatorially equivalent rectangle, we obtain $m$ distinct coordinate values in each dimension.}
		
		For each $i \in [k]$, let $X_i$ be a copy of $X$ and let $\fS_i$ be a copy of $\fS$, both shifted by $(i-1)m$ in the first coordinate. Furthermore, for each point $p \in X$, consider the hyperrectangle {$R_p = \{ q \mid 1 \le q_1 \le km, \forall i \in [d] \setminus \{1\}: q_i = p_i \}$} that contains all copies of $p$. Now let $X' = \bigcup_{i=1}^k X_i$ and $\fS' = \{ R_p \mid p \in X \} \cup \bigcup_{i=1}^k \fS_i$. Clearly, $|X'| = km,$ and $|\fS'| = kn+ m$. It remains to show that $(X, \fS)$ is $(t+1)$-disjunct.
		
		Let $p \in X'$ and $Y \subseteq X'$ with $|Y| = t+1$. We {claim} that there is some rectangle $R \in \fS'$ such that $p \in R$ and $Y \cap R = \emptyset$. This implies that $\fS[p] \not\subseteq \fS[Y]$.
		
		Recall that there is some $q \in X$ such that $p$ is a copy of $q$, i.e.\ $p \in R_q \cap X'$. First assume that $R_q \cap Y = \emptyset$. Then the claim is true simply for $R = R_q$. Otherwise, let $s \in R_q \cap Y$. Let $i \in [k]$ such that $p \in X_i$. By construction, $s \in R_q \setminus \{p\}$ implies that $s \notin X_i$. Let $Y' = Y \cap X_i \subseteq Y \setminus \{s\}$, and observe that $|Y'| \le t$. By assumption, $(X_i, \fS_i)$ is $t$-disjunct, which means that there is some rectangle $R \in \fS_i$ for which $p \in R$ and $Y' \cap R = \emptyset$. Moreover, $R \cap X \subseteq X_i$ by construction, so $Y \cap R = \emptyset$. This proves the claim.
	\end{proof}
	
	\begin{corollary} \label{p:long_rects_step}
		Let $c \ge 1$ and $\rD^d_t(n) \in \Omega(n^c)$. Then $\rD^d_{t+1}(n) \in \Omega(n^{2-\sfrac{1}{c}})$.
	\end{corollary}
	\begin{proof}
		For some constant $c'>0$ we have $\rD_t^d(n) \ge c'n^c$ for all large enough $n$. Then, for each such $n$, there is a set $\fS$ of $d$-rectangles and a set $X \subseteq \R^d$ of points for which $(X, \fS)$ is $t$-disjunct and $|\fS| = n$, and $|X| = \ceil{c'n^c}$. Choosing $k = \ceil{n^{c-1}}$, \Cref{p:long_rects_constr} yields a $(t+1)$-disjunct set system $(X', \fS')$ such that
		\begin{align*}
			& |\fS'| = k \cdot |\fS| + |X| \le n^c + n + c'n^c \le (2+c')n^c, \text{ and} \\
			& |X'| = k \cdot |X| \ge c'n^{2c-1}.
		\end{align*}
		As such, for some $c''>0$ and large enough $n$, $|X'| \ge c''|\fS'|^{(2c-1)/c} = c''n^{2-\sfrac{1}{c}}$.
	\end{proof}
	
	Using \Cref{p:hyperplane-rect-lower-bounds} with $t = d-1$ as induction base and \Cref{p:long_rects_step} as induction step, we obtain:
	
	\restateLongRectsLowerBounds*
	
	\section{Upper bounds} \label{sec:rect-upper-bounds}
	
	We now present a technique to improve the trivial $\fO(n^d)$ upper bound for all $t, d \ge 2$. Throughout this section, we always assume that in a hyperrectangle-induced set system $(P, \fS)$, all points and rectangle corners have integral coordinates in $[4n]$, where $n = |\fS|$, and no hyperplane supporting a rectangle facet contains a point.
	
	To show that this assumption is without loss of generality, we describe an explicit transformation for each dimension $i \in [d]$.
	Recall that we may assume $(P, \fS)$ to be in general position, and thus the rectangles in $\fS$ are defined by $2n$ distinct $i$-coordinates. These coordinates induce $2n-1$ ``strips''. Each point in $P$ must lie in the interior of such a strip, or outside all strips. If two points lie in the same strip (or if both lie in no strip), then move them to align in the $i$-th coordinate (keeping the other coordinates unchanged). Now the points are defined by at most $2n$ distinct $i$-coordinates. After doing this transformation for each $i \in [d]$, the coordinates of points and rectangle corners take at most $4n$ distinct values in each dimension. After another simple coordinate-wise transformation, all coordinates become integers in $[4n]$.
	
	We start with a lemma that allows generalizing bounds to higher dimensions.
	\begin{lemma}\label{p:upper-bound-gen}
		For all constants $d, t, k \in \N_+$, we have: 
		\begin{eqnarray*}
\rS^{d+k}_t(n) & \le & (4n)^k \cdot \rS^d_t(n)\mbox{, and}\\
		\rD^{d+k}_t(n) & \le & (4n)^k \cdot \rD^{d}_t(n).
		\end{eqnarray*}
	\end{lemma}
	\begin{proof}
		We only prove the lemma for $k=1$, as the case $k>1$ follows by induction.
		
		Let $(P, \fS)$ be a $(d+1)$-rectangle-induced set system, and let $n = |\fS|$. Recall that we assume $P \subseteq [4n]^{d+1}$. Let $H$ be an axis-parallel hyperplane in $[4n]^{d+1}$, obtained by fixing the first coordinate to a value in $[4n]$. Let $P' = P \cap H$ and $\fS' = \{ R \cap H \mid R \in \fS \}$.
		
		Our key observation is that if $(P, \fS)$ is $t$-separable ($t$-disjunct), then $(P', \fS')$ is $t$-separable ($t$-disjunct). This implies that $|P'| \le \rS^d_t(|\fS'|) \le \rS^d_t(n)$ (respectively $|P'| \le \rD^d_t(n)$). Finally, we can cover $[4n]^{d+1}$ with $4n$ such hyperplanes $H$, which implies the desired bounds.
	\end{proof}
	
	Observe that \Cref{p:upper-bound-gen} implies that a bound of $o(n^2)$ in two dimensions translates to bounds $o(n^d)$ in all dimensions $d>2$.
	In \Cref{sec:rect-decomp,sec:stab-rects}, we prove that $\fS^2_2(n) \in \fO( n^{\sfrac{5}{3}} )$, which implies that $\fS^d_t(n) \in \fO( n^{d - \sfrac{1}{3}} )$ for all $t, d \ge 2$. We strengthen this result in \Cref{sec:stab-Z-shape} for larger $t$, proving \Cref{p:decomp-sep-upper-bound}. In \Cref{sec:stab-stars}, we adapt our technique to the $t$-disjunct case, proving \Cref{p:decomp-disj-upper-bound}.
	
	Before proceeding with the technical proofs, we give a high-level description of the upper bound for $2$-separability in two dimensions. Let $(P, \fS)$ be a $2$-separable rectangle-induced set system, and let $n = |\fS|$. Assume that the coordinates of $P$ are positive integers in $[4n]$. We want to show that $|P| \in \fO( n^{\sfrac{5}{3}} )$.
	
	An \emph{induced rectangle} in $P$ is a set of four points in $P$ that form the corners of an axis-parallel rectangle (which is not necessarily in $\fS$). It is known that if $P$ contains no induced rectangle, then $|P| \in \fO(n^{\sfrac{3}{2}})$. (This is a special case of the well-known \emph{Zarankiewicz problem}~\cite{KovariSosEtAl1954}.) We can thus assume that an induced rectangle $I \subseteq P$ exists. Let $P_1, P_2 \subset P$ be the two sets of opposing corners of $I$. As $(P, \fS)$ is $2$-separable, there must be some rectangle $R \in \fS$ with $R \cap P_1 \neq \emptyset$ and $R \cap P_2 = \emptyset$, or vice versa. It is easy to see that this means that the rectangle formed by $I$ contains a corner of $R$ (see \Cref{fig:intersect-induced-rect}).
	
	\begin{figure}[htbp]
		\centering \small
		\begin{tikzpicture}[
			scale = 0.4,
			point/.style = {circle, fill, inner sep = 1.2pt}
		]
		\newcommand{\padding}{0.5}
		\begin{scope}
			\draw[gray] (0,0) grid (8,5);
			\node[point, blue] at (1,1) {};
			\node[point, blue] at (7,4) {};
			\node[point, red] at (1,4) {};
			\node[point, red] at (7,1) {};
			\draw (1-\padding, 1-\padding) rectangle (7+\padding,4+\padding);
			\node at (4, -2) {(a)};
		\end{scope}
		\begin{scope}[shift = {(10, 0)}]
			\draw[gray] (0,0) grid (8,5);
			\node[point, blue] at (1,1) {};
			\node[point, blue] at (7,4) {};
			\node[point, red] at (1,4) {};
			\node[point, red] at (7,1) {};
			\draw (1-\padding, 1-\padding) rectangle (7+\padding,1+\padding);
			\node at (4, -2) {(b)};
		\end{scope}
		\begin{scope}[shift = {(20, 0)}]
			\draw[gray] (0,0) grid (8,5);
			\node[point, blue] at (1,1) {};
			\node[point, blue] at (7,4) {};
			\node[point, red] at (1,4) {};
			\node[point, red] at (7,1) {};
			\draw (1-\padding, 1-\padding) rectangle (1+\padding,1+\padding);
			\node at (4, -2) {(c)};
		\end{scope}
		\end{tikzpicture}
		\caption{The three ways a rectangle $R$ can intersect an induced rectangle $I$, up to rotation. One intersection of type (c) is necessary, as otherwise each rectangle that contains a blue point also contains a red point, and vice versa.}
		\label{fig:intersect-induced-rect}
	\end{figure}
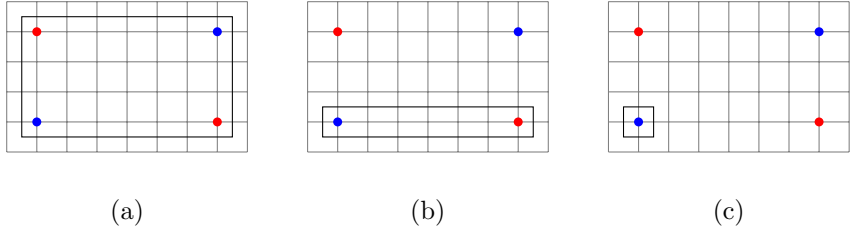
	
	Let $V$ be the set of corners of the rectangles in $\fS$. Our observations imply that each induced rectangle in $P$ must contain a point of $V$; we say that $V$ \emph{stabs} $P$. It is now sufficient to prove that if $V$ stabs $P$, then $|P| \in \fO(|V|^{\sfrac{5}{3}})$ (observe that $|V| \le 4 n$).
	
	Towards this claim, we cover the square $[4n]^2$ (which contains both $P$ and $V$) with rectangles $\fQ$ such that no rectangle $Q \in \fQ$ contains a point from $V$ in its interior. This means that the points of $P$ that fall within $Q$ cannot induce a rectangle. Again invoking the Zarankiewicz problem, we can bound the number of points in $P \cap Q$. Choosing a good rectangle covering (in a non-trivial way) finally yields the desired bound.
	
	\subsection{Hyperrectangle coverings}\label{sec:rect-decomp}
	
	We start with the last part of the high-level description, generalized to higher dimensions. We consider the resulting extremal problem interesting on its own.
	
	A \emph{hyperrectangle covering} of the grid $G = [n]^d$ is a set $\fQ$ of $d$-rectangles (with integral corners) such that $G \subseteq \bigcup \fQ$. We say that $\fQ$ is \emph{valid} with respect to some set $V \subseteq G$ of points if $V \cap \interior(Q) = \emptyset$ for all $Q \in \fQ$. Given a set of $\fO(n)$ points $V$ and a \emph{weight function} $w \colon \fR_d \rightarrow \R$, we ask for the minimum total weight $W(\fQ) = \sum_{Q \in \fQ} w(Q)$ that a valid hyperrectangle covering $\fQ$ can have.
	
	In this section, we define a suitable weight function for the case $d = t = 2$.
	(In the following sections, we will use different weight functions.) For a rectangle $Q$, let the weight $z(Q)$ be the maximum number of integral points that we can fit in $Q$ without inducing an axis-parallel rectangle.
	The total weight of a rectangle covering $\fQ$ is denoted by $Z(\fQ)$. 

	In the following, we say that a rectangle (with integral corners) is a \emph{size $p \times q$} rectangle if it covers a $p \times q$ grid, i.e.\ has side lengths $p-1$ and $q-1$. The K\H{o}v\'{a}ri-S\'{o}s-Tur\'{a}n theorem implies the following:
	\begin{lemma}[\cite{KovariSosEtAl1954,Hylten-Cavallius1958}]\label{p:KovariSos}
		A $p \times q$ rectangle $Q$ where $p \ge q$ has weight  %
		$z(Q) \le p^{\sfrac{1}{2}} q + p$.
	\end{lemma}
	
	Note that among rectangles with the same area, a $k^2 \times k$ rectangle is the ``lightest''. We use this observation in the following upper bound.
	
	\begin{figure}
		\centering
		\newcommand{\n}{27}
		\begin{tikzpicture}[
		scale = 6/(\n),
		decompline/.style={blue, line width = 1pt},
		point/.style = {circle, fill, inner sep = 1.2pt}
		]
		\draw[lightgray] (1,1) grid (\n, \n);
		\draw[decompline] (1, 1) rectangle (9,3);
\draw[decompline] (1, 4) rectangle (9,6);
\draw[decompline] (1, 7) rectangle (9,9);
\draw[decompline] (1, 10) rectangle (9,12);
\draw[decompline] (1, 13) rectangle (9,15);
\draw[decompline] (1, 16) rectangle (9,18);
\draw[decompline] (1, 19) rectangle (9,21);
\draw[decompline] (1, 22) rectangle (9,24);
\draw[decompline] (1, 25) rectangle (9,27);
\draw[decompline] (10, 1) rectangle (18,3);
\draw[decompline] (10, 4) rectangle (18,6);
\draw[decompline] (10, 7) rectangle (18,9);
\draw[decompline] (10, 10) rectangle (18,12);
\draw[decompline] (10, 13) rectangle (18,15);
\draw[decompline] (10, 16) rectangle (18,18);
\draw[decompline] (10, 19) rectangle (18,21);
\draw[decompline] (10, 22) rectangle (18,24);
\draw[decompline] (10, 25) rectangle (18,27);
\draw[decompline] (19, 1) rectangle (27,3);
\draw[decompline] (19, 4) rectangle (27,6);
\draw[decompline] (19, 7) rectangle (27,9);
\draw[decompline] (19, 10) rectangle (27,12);
\draw[decompline] (19, 13) rectangle (27,15);
\draw[decompline] (19, 16) rectangle (27,18);
\draw[decompline] (19, 19) rectangle (27,21);
\draw[decompline] (19, 22) rectangle (27,24);
\draw[decompline] (19, 25) rectangle (27,27);
\draw[decompline](3, 13) -- (3,15);
\draw[decompline](4, 1) -- (4,3);
\draw[decompline](6, 16) -- (6,18);
\draw[decompline](16, 16) -- (16,18);
\draw[decompline](17, 4) -- (17,6);
\draw[decompline](22, 19) -- (22,21);
\draw[decompline](23, 25) -- (23,27);
\draw[decompline](26, 22) -- (26,24);
\node[point] at (1, 13) {};
\node[point] at (2, 25) {};
\node[point] at (3, 14) {};
\node[point] at (4, 2) {};
\node[point] at (5, 9) {};
\node[point] at (6, 17) {};
\node[point] at (7, 16) {};
\node[point] at (8, 13) {};
\node[point] at (9, 26) {};
\node[point] at (10, 27) {};
\node[point] at (11, 10) {};
\node[point] at (12, 16) {};
\node[point] at (13, 12) {};
\node[point] at (14, 19) {};
\node[point] at (15, 7) {};
\node[point] at (16, 17) {};
\node[point] at (17, 5) {};
\node[point] at (18, 10) {};
\node[point] at (19, 5) {};
\node[point] at (20, 25) {};
\node[point] at (21, 4) {};
\node[point] at (22, 20) {};
\node[point] at (23, 26) {};
\node[point] at (24, 9) {};
\node[point] at (25, 18) {};
\node[point] at (26, 23) {};
\node[point] at (27, 26) {};
		\end{tikzpicture}
		\caption{A two-dimensional rectangle covering as in \Cref{p:rect-cover-Zar} with $k = 3$, where rectangles are always divided with vertical lines. Blue lines indicate rectangle borders.}
		\label{fig:rect-decomp}
	\end{figure}
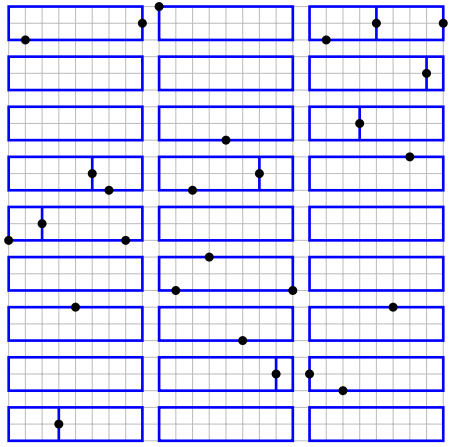
	
	\begin{lemma} \label{p:rect-cover-Zar}
		Let $n = k^3$ for some $k \in \N_+$, and let $V \subseteq [n]^2$ such that $|V| \le n$. Then there is a rectangle covering $\fQ$ with $Z(\fQ) \le 4 n^{\sfrac{5}{3}}$ that is valid with respect to $V$.
	\end{lemma}
	\begin{proof}
		We construct $\fQ$ as follows. Start with the regular decomposition into $n$ rectangles of size $k^2 \times k$, i.e.

		\begin{align*}
			\fQ_0 = \{ [(i-1) k^2+1, i k^2] \times [(j-1) k+1, j k] \mid i \in [k], j \in [k^2] \}.
		\end{align*}
		Observe that while these rectangles do not cover the \emph{continuous} square $[1,n]^2$, they do cover all integral points in $[n]^2$. Moreover, each rectangle in $\fQ_0$ has weight at most $2k^2$, so $Z(\fQ_0) = n \cdot 2 k^2 = 2 n^{\sfrac{5}{3}}$.
		
		The obtained $\fQ_0$ is not yet a valid rectangle covering, as its rectangles may contain points of $V$ in their interiors. To fix this, for each point $v \in V$ and rectangle $Q \in \fQ_0$ where $v \in \interior(Q)$, split $Q$ with an axis-parallel line through $v$ (see \Cref{fig:rect-decomp}, where we always split the rectangles with vertical lines). Observe that these splits increase the number of rectangles by at most $|V|$. Let $\fQ$ be the set of rectangles thus obtained. As the weight of a rectangle obtained by splitting can only be smaller than the weight of the original rectangle, we have $z(Q) \le 2k^2$ for all $Q \in \fQ$, and thus,
		\begin{align*}
			Z(\fQ) \le Z(\fQ_0) + |V| \cdot 2k^2 \le 4n^{\sfrac{5}{3}}. \tag*{\qedhere}
		\end{align*}
	\end{proof}
	
	We briefly argue that \Cref{p:rect-cover-Zar} is asymptotically tight. It is known that \Cref{p:KovariSos} is asymptotically tight (see e.g.\ \cite{Reiman1958}). In the following, we ignore constants and assume $z(q) = p^{\sfrac{1}{2}} q + p$. Given a rectangle with side lengths $p \ge q$, let the \emph{density} of $Q$ be defined as $\delta(Q) = z(Q) / (pq) = p^{-\sfrac{1}{2}} + q^{-1}$, and let the density of a rectangle covering $\fQ$ of $[n]^2$ be defined as $\Delta(\fQ) = Z(\fQ) / A(\fQ)$, where $A(\fQ)$ denotes the total area of all rectangles in $\fQ$. Observe that $A(\fQ) \geq n^2$, with equality if the rectangles do not overlap. As $\Delta(\fQ)$ is the weighted average of the densities of the individual rectangles, we have
	\begin{align*}
		\Delta(\fQ) \ge \min_{Q \in \fQ} \delta(Q).
	\end{align*}
	
	Let $n = k^3$ and consider the set $V \subseteq [n]^2$ of points shown in \Cref{fig:rect-cover-hard}. It consists of a large $k \times k$ grid, where each grid cell (of size $k^2 \times k^2$) contains a set of $k$ evenly spaced points on its diagonal. Let $Q$ be a $p \times q$ rectangle such that $\interior(Q) \cap V = \emptyset$. Observe that if $p > k^2$, then $q \le k$. As such, $\delta(Q) \ge \frac{1}{k}$. We now have $Z(Q) = A(\fQ) \cdot \Delta(\fQ) \geq n^2 \cdot \Delta(\fQ) \ge n^2 \cdot \frac{1}{k} = n^{\sfrac{5}{3}}$.
	
	\begin{figure}[htbp]
		\centering
		\small
		\begin{tikzpicture}[
			scale=5,
			point/.style = {circle, fill, inner sep = 1pt},
			bounds/.style={},
			lines/.style={dash pattern={on 2pt off 2pt}}
		]
			\input{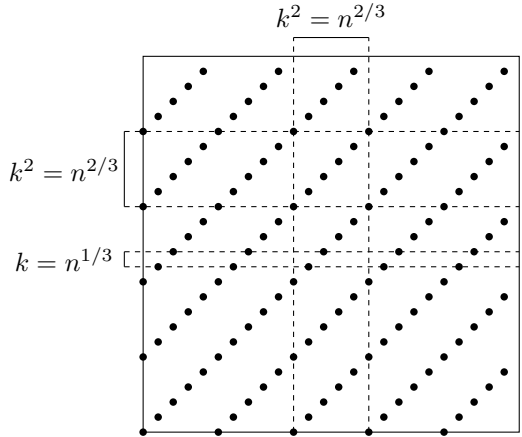}
		\end{tikzpicture}
		\caption{A set of $n=k^3$ points $V$ that requires a rectangle covering of weight $n^{\sfrac{5}{3}}$.}
		\label{fig:rect-cover-hard}
	\end{figure}
	
	\subsection{Stabbing induced rectangles}\label{sec:stab-rects}
	
	Let $V, P \subseteq [n]^2$ be point sets. We say that $P$ \emph{induces} a rectangle $R$ if all four corners of $R$ are in $P$, and we say that $V$ \emph{stabs} $P$ if every rectangle induced by $P$ contains a point from $V$ in its interior.
	
	\begin{lemma} \label{p:stabbing}
		Let $n = k^2$ for some $k \in \N$, and let $V, P \subseteq [n]^2$ such that $|V| \le n$ and $V$ stabs $P$. Then $|P| \le 4 n^{\sfrac{5}{3}}$.
	\end{lemma}
	\begin{proof}
		\Cref{p:rect-cover-Zar} implies that there is a rectangle covering $\fQ$ with $Z(\fQ) \le 4n^{\sfrac{5}{3}}$ that is valid with respect to $V$. Now assume that $|P| > Z(\fQ)$. By the pigeonhole principle, there is a rectangle $Q \in \fQ$ that contains more than $z(Q)$ points. By the definition of $z$, inside $Q$ there must be $4$ points of $P$ that induce a rectangle. The interior of this induced rectangle contains no point from $V$, contradicting the assumption that $V$ stabs $P$. Thus, $|P| \le Z(\fQ)$.
	\end{proof}

	Note that while \Cref{p:rect-cover-Zar} is tight, we do not know whether \Cref{p:stabbing} is also tight. The lower bound for rectangle coverings does not simply transfer, because if we construct $P$ by ``filling up'' all rectangles of the covering, then two pairs of points contained in different rectangles may still induce a rectangle that is not stabbed.
	
	We are now ready to prove \Cref{p:decomp-sep-upper-bound-two-dim}.
	
	\restateDecompSepUpperBoundTwo*
	\begin{proof}
		We only show the case $d = 2$, the other cases follow by \Cref{p:upper-bound-gen}.
		
		Let $(P, \fS)$ be a 2-separable rectangle-induced set system. Let $n = |\fS|$ and let $V$ be the set of $n' = 4n$ corners of rectangles in $\fS$. We assume that $n'$ is of the form $k^3$, that $P, V \subseteq [n']$, and that points $p \in P$ and $v \in V$ differ in all coordinates.
		
		Let $P' \subseteq P$ be a set of 4 points that induce a rectangle $R$. Let $P_1$ and $P_2$ be the pairs of opposite corners of $R$. We claim that there must be a rectangle $S \in \fS$ that contains only one point of $P$. Suppose not, then each rectangle $S \in \fS$ either contains no point of $P$ or at least one ``edge'' of $R$, and thus both a point from $P_1$ and a point from $P_2$ (see \Cref{fig:intersect-induced-rect}). This implies $\fS[P_1] = \fS[P_2]$, contradicting that $(P, \fS)$ is 2-separable.
		
		As such, $R$ contains a point of $V$, which is in the interior of $R$ by assumption. As a consequence, $V$ must stab $P$, and \Cref{p:stabbing} implies $|P| \le 4 (n')^{\sfrac{5}{3}} \in \fO( n^{\sfrac{5}{3}} )$.
	\end{proof}

	\subsection{Stabbing Z-shaped point sets}\label{sec:stab-Z-shape}
	
	We slightly modify the technique of the previous sections and obtain improved bounds in the case $t = 3$. Informally, instead of the avoidance of \emph{induced rectangles}, we argue about avoidance of a different pattern, called the \emph{Z-shape} (\Cref{zshape}). More formally, call a set of four distinct points $w,x,y,z \in \N^2$ a \emph{Z-shape} if $w_1 < x_1 = y_1 < z_1$ and $w_2 = x_2 < y_2 = z_2$.
	
	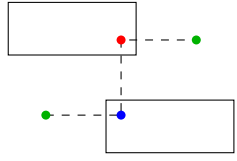
\begin{figure}[htbp]
		\begin{center}
			\begin{tikzpicture}[
				scale=1,
				point/.style = {circle, fill, inner sep = 1.2pt},
				inAB/.style = {green!70!black}
			]
				\draw[dashed] (0,0) -- (1,0) -- (1,1) -- (2,1);
				\node[point, inAB] at (0,0) {};
				\node[point, blue] at (1,0) {};
				\node[point, red] at (1,1) {};
				\node[point, inAB] at (2,1) {};
				\draw (0.8,0.2) rectangle (2.5,-0.5);
				\draw (1.2,0.8) rectangle (-0.5,1.5);
			\end{tikzpicture}
		\end{center}
		\caption{A Z-shaped point set. The coloring indicates the sets $X_1, X_2$ used in the proof of \Cref{p:decomp-sep-upper-bound}: red, blue, and green points are in $X_1 \setminus X_2$, $X_2 \setminus X_1$, resp. $X_1 \cap X_2$.\label{zshape}}
	\end{figure}
	
	We now adapt \Cref{p:rect-cover-Zar}. For a rectangle $Q$ of size $m \times n$, let $z'(Q)$ denote the maximum number of points that fit into $Q$ without creating a Z-shaped subset. F\"{u}redi and Hajnal~\cite[Cor.\ 2.4]{FuerediHajnal1992} showed (using different notation) that $z'(Q) \in \fO( m + n )$.
	For a rectangle covering $\fQ$, let $Z'(\fQ) = \sum_{Q \in \fQ} z'(Q)$.
	
	\begin{lemma} \label{p:rect-cover-Z-shape}
		Let $V \subseteq [n]^2$ such that $|V| \le n$. Then there is a rectangle covering $\fQ$ that is valid with respect to $V$, such that $Z'(\fQ) \in \fO( n^{\sfrac{3}{2}})$.
	\end{lemma}
	\begin{proof}
		Assume that $n = k^2$ for some $k \in \N_+$. (We can ensure this by padding, without affecting the asymptotic bounds.) %
		
		Start with the regular decomposition $\fQ_0$ into $n$ squares of size $k \times k$. We have $|\fQ_0| = n$ and $z'(Q) \in \fO( k ) = \fO( \sqrt{n} )$ for each $Q \in \fQ$.
		
		Again, for each $v \in V$ in the interior of some rectangle $Q \in \fQ_0$, split $Q$ with an axis-parallel line through $v$. This yields a decomposition $\fQ$ with at most $2n$ rectangles of size at most $k \times k$. We thus have
		\begin{align*}
			Z'(\fQ) \in \fO( n \cdot \sqrt{n} ) = \fO( n^{\sfrac{3}{2}}). \tag*{\qedhere}
		\end{align*}
	\end{proof}
	
	We proceed with adapting our notion of stabbing. For a Z-shaped point set $X$, let $\rect(X)$ be the smallest axis-parallel rectangle that contains $X$. We say that a point set $V$ \emph{Z-stabs} another point set $P$ if for every Z-shaped subset $X \subseteq P$, we have $\interior(\rect(X)) \cap V \neq \emptyset$. We can now easily adapt \Cref{p:stabbing}, using \Cref{p:rect-cover-Z-shape}.
	
	\begin{lemma}\label{p:Z-stabbing}
		Let $V, P \in [n]^2$ such that $|V| \le n$ and $V$ Z-stabs $P$. Then $|P| \in \fO( n^{\sfrac{3}{2}})$.
	\end{lemma}
	
	We are now ready to prove the general upper bound:
	
	\restateDecompSepUpperBoundHigher*
	\begin{proof}
		We only need to show the case $d = 2$ and $t = 3$, as the other cases follow by \Cref{p:upper-bound-gen} and monotonicity (\Cref{p:monotone}).
		
		Let $(P, \fS)$ be a 3-separable rectangle-induced set system. Let $n = |\fS|$ and let $V$ be the set of $n' = 4n$ corners of rectangles in $\fS$. We assume that $P, V \subseteq [n']$ and that points $p \in P$ and $v \in V$ differ in all coordinates.
		
		Let $X = \{w,x,y,z\} \subseteq P$ be a Z-shaped subset such that $w_1 < x_1 = y_1 < z_1$ and $w_2 = x_2 < y_2 = z_2$. Let $X_1 = \{w,x,z\}$ and $X_2 = \{w,y,z\}$. As $(P, \fS)$ is 3-separable, there must be a rectangle $S \in \fS$ that intersects exactly one of the two sets $X_1$ and $X_2$. It is easy to see that this means that a corner of $S$ is in $\rect(X)$. See \Cref{zshape} for an illustration.
		
		It follows that $V$ Z-stabs $P$, and \Cref{p:Z-stabbing} yields the desired bound.
	\end{proof}

	\subsection{Stabbing stars}\label{sec:stab-stars}
	
	For the $d$-disjunct case, we can prove stronger bounds with essentially the same technique, but using a different weight function and hyperrectangle coverings in arbitrary dimensions.
	
	Let $T = \{ x, x_1, x_2, \dots, x_d \} \subseteq \N^d$ be a set of $d+1$ points such that for all $i \in [d]$, $x_i$ differs from $x$ only in the $i$-th coordinate. Then we call $T$ a $d$-\emph{star}, and refer to $x$ as the \emph{center} of $T$. (In other words, a $d$-star consists of a corner of a $d$-rectangle and all its adjacent corners.) Let $\rect(T)$ denote the unique rectangle that contains all points in $T$ as corners.
	
	Given a $d$-rectangle $Q$, let $a_d(Q)$ denote the maximum number of integral points that fit in $Q$ without inducing a $d$-star. Let $A_d(\fQ)$ denote the total weight of a hyperrectangle covering $\fQ$. We bound $a_d(Q)$ for hypercubes.
	
	\begin{lemma} \label{p:semi-rect-free}
		Let $P \subseteq [n]^d$ contain no $d$-star. Then $|P| \le d n^{d-1}$.
	\end{lemma}
	\begin{proof}
		Consider a point $p \in P$. There must be some axis-parallel line through $P$ that contains no other points (otherwise $p$ would be the center of a star). Thus, $|P|$ is at most the number of axis-parallel lines that intersect $P$, which is $d n^{d-1}$.
	\end{proof}
	
	We now adapt \Cref{p:rect-cover-Zar} to the weight function $a_d$.
	
	\begin{lemma} \label{p:hyperrect-cover-star}
		Let $k, d \in \N$, let $n = k^d$ and let $V \subseteq [n]^d$ such that $|V| \le cn$ for some constant $c \ge 1$. Then there is a hyperrectangle covering $\fQ$ that is valid with respect to $V$ such that $A_d(\fQ) \le (1+c) d n^{d - 1 + \sfrac{1}{d}}$.
	\end{lemma}
	\begin{proof}
		For convenience, denote $m = n/k = n^{1-\sfrac{1}{d}}$. Start with the regular decomposition of $[n]^d$ into $n$ \emph{hypercubes}, i.e.
		\begin{align*}
			& \fQ_0 = \{ Q_{i_1, i_2, \dots, i_d} \mid 0 \le i_j < m-1 \text{ for } j \in [d] \}, \text{ where} \\
			&  Q_{i_1, i_2, \dots, i_d} = [i_1 m+1,(i_1+1)m] \times [i_2 m+1,(i_2+1)m] \times \dots \times [i_d m+1,(i_d+1)m].
		\end{align*}
		The weight of one such cube is $m^{d-1} \le d n^{d-2+\sfrac{1}{d}}$, so $A_d(\fQ_0) \le d n^{d-1+\sfrac{1}{d}}$.
		
		For each $v \in V$ that is not yet on a rectangle boundary, we split the containing rectangle with an axis-parallel hyperplane through $v$. This adds no more than $|V| \le cn$ rectangles of weight at most $m^{d-1}$.
	\end{proof}
	
	Now let $V, P \subseteq [n]^d$. We say that $V$ \emph{semi-stabs} $P$ if for every $d$-star $T \subseteq P$, the rectangle $\rect(T)$ contains a point of $V$ in its interior. Adapting \Cref{p:stabbing} and using \Cref{p:hyperrect-cover-star} instead of \Cref{p:rect-cover-Zar} immediately yields the following.
	
	\begin{lemma} \label{p:semi-stabbing}
		Let $n = k^d$ for some $k \in \N$, and let $V, P \subseteq [n]^d$ such that $V$ semi-stabs $P$ and $|V| \le c n$ for some constant $c \ge 1$. Then
		\begin{align*}
			|P| \le (1+c) d n^{d-1+\sfrac{1}{d}}.
		\end{align*}
	\end{lemma}
	
	Finally, observe that if $\fS \subseteq \fR_d$ and $(P, \fS)$ is $d$-disjunct, then each star $T \subseteq P$ must be stabbed by some corner of a hyperrectangle in $\fS$. In particular, there must be a rectangle that contains only the center of $T$.
	
	Let $n = 4|\fS|$. There are $2^d |\fS| = 2^{d-2} n$ hyperrectangle corners. Thus, \Cref{p:semi-stabbing} with $c = 2^{d-2} \in \fO(1)$ implies $\rD_d^d(n) \in \fO(n^{d-1+\sfrac{1}{d}})$. Applying \Cref{p:upper-bound-gen} yields $\rD_t^d(n) \in \fO( n^{d-1+\sfrac{1}{t}})$ for $t \le d$ and applying monotonicity of $\rD_t^d(n)$ in $t$ yields $\rD_t^d(n) \in \fO( n^{d-1+\sfrac{1}{d}})$ for $t \ge d$. Together, we obtain:
	
	\restateDecompDisjUpperBound*

	\section{Grids with axis-parallel affine subspaces}\label{sec:grid-subspaces}
	
	Finally, we consider another generalization of the grid-line construction in \Cref{sec:grid-lines}. For $k, d, m \in \N$, let $P_d = [m]^d$ and let $\fP_{k,d}$ be the set of $k$-dimensional axis-parallel affine subspaces of $\R^d$ that intersect some point in $P_d$. In other words, $\fP_{k,d}$ is the collection of subsets of $P_d = [m]^d$ that we obtain by fixing $d-k$ coordinates to some values in $[m]$. Note that $\fL_d = \fP_{1,d}$. We first prove $(d-k)$-disjunctness of $(P_d, \fP_{k,d})$, then show how to construct an equivalent $(d-1)$-rectangle-induced set system. This slightly improves \Cref{p:hyperplane-rect-lower-bounds} in some cases.
	
	In the following, for a set $I \subseteq [d]$ and a point $p \in \R^d$, let $H_I(p)$ be the axis-parallel affine subspace that extends in the coordinates $I$ and contains $p$, i.e.
	\begin{align*}
		H_I(p) = \{ x \in \R^d \mid \forall i \notin I : x_i = p_i \}.
	\end{align*}
	
	\begin{lemma}
		$(P_d, \fP_{k,d})$ is $(d-k)$-disjunct.
	\end{lemma}
		
	\begin{proof}
		Suppose $(P_d, \fP_{k,d})$ is not $(d-k)$-disjunct. Then there is some 
		 $a \in P_d$ and $B \subseteq P_d \setminus \{a\}$ such that $\fP_{k,d}[a] \subseteq \fP_{k,d}[B]$ and $|B| \leq d-k$.		
		
		Let $b \in B$ and let $J(b)$ be the set of coordinates in which $a$ and $b$ differ. 
Observe that $\fP_{k,d}[a] = \{ H_I(a) \mid I \subseteq [d], |I| = k \}$, and $b \in H_I(a)$ if and only if $J(b) \subseteq I$.

We first observe that we may assume that all points in $B$ differ from $a$ in exactly one coordinate. Indeed, suppose there is some $b \in B$ such that $|J(b)| > 1$. Let $b' \in B$ be a point that differs from $a$ in an arbitrary single coordinate $j \in J(b)$. Then $J(b') = \{j\} \subset J(b) \subseteq I$, so point $b'$ hits all the elements of $\fP_{k,d}[a]$ that $b$ hits, so we can replace $b$ by $b'$ in $B$. We can iteratively replace all points of $B$ in this way, until the assumption holds. (The size of $B$ may only go down via this process.)

		Consequently, there are $d - |B| \geq d-(d-k) = k$ coordinates on which all points of $B$ agree with $a$. Let $I$ be such a set of coordinates of size $k$. Then %
		$H_I(a) \in \fP_{k,d}[a]$ and $H_I(a) \notin \fP_{k,d}[B]$, a contradiction.
	\end{proof}
	
	Observe that $|P_d| = m^d$ and $|\fP_{k,d}| = \binom{d}{k} m^{d-k}$. As axis-parallel affine subspaces are equivalent to rectangles for our purposes, plugging in $k = d-t$ already yields $\rD_t^d(n) \in \Omega( n^{d/t} )$. This bound, however, is at most as good as \Cref{p:hyperplane-rect-lower-bounds} in all cases. We strengthen it by reducing the dimension.
	
	\begin{lemma}\label{p:subspace-proj}
		For each $k \le d-2$, there is a $(d-1)$-rectangle-induced set system $(Q, \fS)$ that is combinatorially equivalent to $(P_d, \fP_{k,d})$.
	\end{lemma}
	\begin{proof}
		Let $\varepsilon = \frac{1}{n+1}$ and let $f : P_d \rightarrow \R^{d-1}$ be defined as follows:
		\begin{align*}
			f(x_1, x_2, \dots, x_d) = (x_1 + \varepsilon x_d, x_2 + \varepsilon x_d, \dots, x_{d-1} + \varepsilon x_d).
		\end{align*}
		Observe that $f(x)_i \le f(x')_i$ if and only if $x_i < x'_i$ or $x_i = x'_i \wedge x_d \le x'_d$.
		
		Similarly to the proof of \Cref{p:grid-line-lower-bounds}, we set $Q = f(P_d)$ and let $\fS = \{ g(H) \mid H \in \fP_{k,d} \}$, where $g(H) = \rect( f(H \cap P_d) )$. We {claim} that $g(H) \cap Q \subseteq f(H \cap P_d)$. This implies equivalence of $(Q, \fS)$ and $(P_d, \fP_{k,d})$.
		
		To prove the claim, let $H_I(a) \in \fP_{k,d}$ and let $x \in P_d$ such that $f(x) \in g(H_I(a))$. We need to show that $x \in H_I(a)$, that is, $x_i = a_i$ for all $i \in [d] \setminus I$.
		
		First, consider some $i \in [d-1] \setminus I$. We know that there are $b, c \in H_I(a)$ such that $f(b)_i \le f(x)_i \le f(c)_i$. As $b_i = c_i = a_i$, our observation above directly implies $x_i = a_i$ (the only other possibility is $b_i < x_i < c_i$, a contradiction).
		
		It remains to show the claim for $i = d$, if $d \notin I$. We know that $|I| = k \le d - 2$, so there is some $j \in [d-1] \setminus I$. Consider some $b, c \in H_I(a)$ such that $f(b)_j \le f(x)_j \le f(c)_j$. Our observation implies that $x_j = a_j$ and, in particular, $x_d = a_d$. This concludes the proof of the claim.
	\end{proof}
	
	Using \Cref{p:subspace-proj}, we finally obtain $D_t^{d-1}(n) \in {\Omega( n^{d/t} )}$, and thus:
	
	\restateGridSubspacesLowerBound*
	
	\section{Conclusion}
	
	We have shown several bounds on the density of $t$-separable and $t$-disjunct set systems induced by axis-parallel boxes. In the case $t=1$, we described an optimal construction, showing that the upper bound cannot be asymptotically improved (Theorem~\ref{p:hyperplane-rect-lower-bounds}). For all non-trivial cases, namely if $t, d \ge 2$, our lower bounds improve upon the trivial $\Omega(n)$ and our upper bounds improve upon the trivial $\fO(n^d)$. Still, apart from special cases, there are large gaps between the current upper and lower bounds, which provides an interesting target for future research. As a particular challenge, consider the planar, $2$-separable case, where we have $c \cdot n^{\sfrac{3}{2}} \leq \rS_2^2(n) \leq c' \cdot n^{\sfrac{5}{3}}$, for some $c,c'>0$.
	
	Intuitively, a larger value of $t$ makes the group testing problem more difficult, so $\rS_{t+1}^d(n) \leq \rS_t^d(n)$, as stated in Proposition~\ref{p:monotone}. It seems likely that a stronger (asymptotic) separation holds.  %
	This question remains open even in the planar case. %

	\begin{openQuestion}
		Is $\rS_{t+1}^2(n) \in o(\rS_t^2(n))$ for all $t \geq 2$?
	\end{openQuestion}

Another intriguing question is the following.

	\begin{openQuestion}
 Is $\rD_{t}^d{(n)} \in o(\rS_t^d(n))$ for some $t,d \geq 2$?
	\end{openQuestion}
	
	As discussed in \Cref{sec:intro}, worst-case configurations of points admit only trivial bounds in our setting. It may be interesting to study the problem for \emph{random} configurations of points.

	\newpage 
	
	\bibliography{ggt}{}
	\bibliographystyle{plain}
	
	\newpage
	
	\section*{Appendix}
	\appendix
	
	\section{At most \texorpdfstring{$t$}{t} defectives}\label{appb}

The restriction to have \emph{exactly} $t$ defective items may seem artificial. Perhaps more natural is to let $t$ be an upper bound on the number of defectives. %
This variant of the problem is also well studied in the group testing literature. We summarize the main results that carry over to our geometric setting.

A set system $(X, \fS)$ is called \emph{$\bar{t}$-separable} if there are no two distinct $Y, Z \subseteq X$ such that $|Y| \leq t$ and $|Z| \leq t$ and $\fS[Y] = \fS[Z]$. Observe that the $\bar{t}$-separability differs from $t$-separability, by allowing the set of defectives to be \emph{smaller} than $t$. Thus $\bar{t}$-separability captures the problem of group testing with \emph{at most} $t$ defectives.

One can similarly modify the definition of $t$-disjunctness, and say that a set system $(X, \fS)$ is $\bar{t}$-disjunct if there is no $Y \subseteq X$ and $x \in X \setminus Y$ such that $|Y| \leq t$ and $\fS[x] \subseteq \fS[Y]$. 

The two notions of disjunctness turn out to be equivalent~\cite{DuHwang1993}. The two notions of separability, however, are different, satisfying the following relations:
\begin{lemma}[Du and Hwang \cite{DuHwang1993}]\label{lem13}
		For each set system $(X, \fS)$ and each $t \ge 1$,
		\begin{align*}
		& \overline{(t+1)}\text{-separable} \implies t\text{-disjunct} \implies \overline{t}\text{-separable}  \implies t\text{-separable}.
		\end{align*}
	\end{lemma}

	We define $\overline{\rS^d_t}(n)$ and $\overline{\rD^d_t}(n)$ to denote the maximum $m$ where there exist a set $X \subseteq \R^d$ of cardinality $m$ and a set $\fS \subseteq \fR_d$ such that $(X, \fS)$ is $\bar{t}$-separable, resp.\ $\bar{t}$-disjunct. \Cref{lem13} implies the following inequalities:
	\begin{align*}
	\overline{\rD^d_t}(n)~~  =  ~~ \rD^d_t(n)~~ \leq ~~ & \overline{\rS^d_t}(n) ~~  \leq  ~~  \rS^d_t(n), \text{ ~~~ and} \\		
		& \overline{\rS^d_t}(n)~~  \leq ~~ \rD^d_{t-1}(n).
\end{align*}

In words, the maximum number of items that $n$ rectangle-tests can handle in $d$ dimensions, with \emph{at most} $t$ defectives is sandwiched between the corresponding quantities with \emph{exactly} $t$ defectives, respectively, with the stronger, disjunctness condition. Alternatively, bounds for disjunctness also hold for the modified notion of separability, up to a possible shift of $t$ by one. 

\newpage
\section{Numerical bounds in special cases}\label{appc}

\begin{figure}[h]
\begin{center}
    \begin{tabular}{ | r | c | c | c | c | c |}
    \hline
     \backslashbox{$t$}{$d$}& $1$ & $2$ & $3$ & $4$ & $5$ \\ \hline
     $1$  & $1$ & $2$ & $3$ & $4$ & $5$ \\ \hline
     $2$  & $1$ & $[1.5,1.67]$ & $[2,2.67]$ & $[2.5,3.67]$ & $[3,4.67]$ \\ \hline
     $3$  & $1$ & $1.5$ & $[2,2.5]$ & $[2,3.5]$ & $[2,4.5]$ \\ \hline
     $4$  & $1$ & $[1.33,1.5]$ & $[1.5,2.5]$ & $[1.5,3.5]$ & $[2,4.5]$ \\ \hline
     $5$  & $1$ & $[1.33,1.5]$ & $[1.33,2.5]$ & $[1.33,3.5]$ & $[1.5,4.5]$ \\ \hline

    \end{tabular}
\end{center}

\begin{center}
    \begin{tabular}{ | r | c | c | c | c | c |}
    \hline
     \backslashbox{$t$}{$d$}& $1$ & $2$ & $3$ & $4$ & $5$ \\ \hline
     $1$  & $1$ & $2$ & $3$ & $4$ & $5$ \\ \hline
     $2$  & $1$ & $1.5$ & $[2,2.5]$ & $[2.5,3.5]$ & $[3,4.5]$ \\ \hline
     $3$  & $1$ & $[1.33,1.5]$ & $[1.5,2.34]$ & $[2,3.34]$ & $[2,4.34]$ \\ \hline
     $4$  & $1$ & $[1.25,1.5]$ & $[1.33,2.34]$ & $[1.5,3.25]$ & $[2,4.25]$ \\ \hline
     $5$  & $1$ & $[1.2,1.5]$ & $[1.2,2.34]$ & $[1.2,3.25]$ & $[1.5,4.2]$ \\ \hline

    \end{tabular}
\end{center}
\caption{Polynomial degree of $\rS_t^d(n)$ (above) and $\rD_t^d(n)$ (below) for small values of $t$ and $d$. Intervals indicate best lower and upper bounds, rounded to two decimal digits. A single number in a cell indicates a tight bound.}
	\end{figure}

\end{document}